\documentclass[onefignum,onetabnum]{siamart190516}

\usepackage{lipsum}
\usepackage{amsfonts}
\usepackage{graphicx}
\usepackage{epstopdf}
\usepackage{algorithmic}
\ifpdf
  \DeclareGraphicsExtensions{.eps,.pdf,.png,.jpg}
\else
  \DeclareGraphicsExtensions{.eps}
\fi

\newsiamremark{remark}{Remark}
\newsiamremark{hypothesis}{Hypothesis}
\crefname{hypothesis}{Hypothesis}{Hypotheses}
\newsiamthm{claim}{Claim}
\newsiamthm{assumption}{Assumption}

\headers{}{}

\title{Convergence of regularized particle filters for stochastic reaction networks \thanks{Submitted to the editors DATE.\funding{This work was funded by the Swiss National Science Foundation under grant number 182653.}}}

\author{Zhou Fang\thanks{Department of Biosystems Science and Engineering, ETH Zurich, Mattenstrasse 26, 4058 Basel, Switzerland
  (\email{zhou.fang@bsse.ethz.ch}, \email{ankit.gupta@bsse.ethz.ch}, \email{mustafa.khammash@bsse.ethz.ch})}
 \and Ankit Gupta\footnotemark[2]
 \and Mustafa Khammash\footnotemark[2]
}

\usepackage{amsopn}

\usepackage{xcolor}
\usepackage{url}
\usepackage{hyperref}
\usepackage{mathtools}

\usepackage{slashbox}
\usepackage{stfloats}
\usepackage{mathbbold}
\usepackage{amssymb}
\usepackage{enumitem}
\usepackage{graphicx}
\usepackage[version=4]{mhchem}
\usepackage{subfigure}

\def\dd{\text{d}}

\ifpdf
\hypersetup{
  pdftitle={},
  pdfauthor={}
}
\fi

\usepackage{xcolor}
\usepackage{url}
\usepackage{hyperref}
\usepackage{mathtools}

\usepackage{slashbox}
\usepackage{stfloats}
\usepackage{mathbbold}
\usepackage{amssymb}
\usepackage{enumitem}
\usepackage{graphicx}
\usepackage[version=4]{mhchem}

\def\dd{\text{d}}

\begin{document}

\maketitle

\begin{abstract}
Filtering for stochastic reaction networks (SRNs) is an important problem in systems/synthetic biology aiming to estimate the state of unobserved chemical species.
A good solution to it can provide scientists valuable information about the hidden dynamic state and enable optimal feedback control. 
Usually, the model parameters need to be inferred simultaneously with state variables,
and a conventional particle filter can fail to solve this problem accurately due to sample degeneracy.
In this case, the regularized particle filter (RPF) is preferred to the conventional ones, as the RPF can mitigate sample degeneracy by perturbing particles with artificial noise.
However, the artificial noise introduces an additional bias to the estimate, and, thus, it is questionable whether the RPF can provide reliable results for SRNs.
In this paper, we aim to identify conditions under which the RPF converges to the exact filter in the filtering problem determined by a bimolecular network. 
First, we establish computationally efficient RPFs for SRNs on different scales using different dynamical models, including the continuous-time Markov process, tau-leaping model, and piecewise deterministic process. 
Then, by parameter sensitivity analyses, we show that the established RPFs converge to the exact filters if all reactions leading to an increase of the molecular population have linearly growing propensities and some other mild conditions are satisfied simultaneously.
This ensures the performance of the RPF for a large class of SRNs, and several numerical examples are presented to illustrate our results.
\end{abstract}



\section{Introduction}

The development of fluorescence technologies \cite{zhang2002creating} and modern microscopes \cite{stephens2003light,vonesch2006colored} in the past few decades has greatly improved scientists' ability to study various biological problems at the single-cell level \cite{locke2009using}.
For instance, researchers can now use fluorescent data to investigate the variability in single-cell gene expression \cite{rosenfeld2005gene}, and the role of transcription factors in shaping bursty dynamics \cite{rullan2018optogenetic}. 
Despite these successes, a time-lapse microscope can only follow a few genes over time due to the limited availability of distinguishable reporters \cite{locke2009using};
consequently, many chemical species in the cell are not directly tracked.
This drawback greatly limits scientists' ability to further investigate and control the dynamical behaviors of single-cell systems.
Consequently, it is of immediate interest to establish efficient stochastic filters that can accurately infer the unobserved species in an intracellular reaction system by partial observations.
 
The filtering theory that aims to calculate the posterior of hidden dynamic states given the time-course observation has been extensively developed in the past few decades.
In the linear and Gaussian scenarios, the solution to the filtering problem is the well-known Kalman filter \cite{kalman1960new}, which can be calculated exactly in a computationally efficient way.
In contrast, the nonlinear and non-Gaussian systems usually lead to infinite-dimensional filters that cannot be calculated explicitly \cite{bain2008fundamentals}.
In this case, the particle filter (PF) originally introduced in \cite{gordon1993novel} is the most effective method to numerically solve the filtering problem, which approximates the posterior by a population of weighted samples generated by importance sampling (from the target dynamics) and resampling \cite{doucet2009tutorial}.
So far, PFs have been successfully applied to biochemical reaction systems under different problem settings, e.g., the heat shock response system \cite{liu2012state}, metabolic dynamics \cite{yang2007vivo}, transcriptional switch models \cite{hey2015stochastic}, multiscale reaction systems \cite{fang2021stochastic,fang2020stochastic}, biochemical processes of noise-free observations \cite{rathinam2020state}.
Moreover, the method has been shown to outperform other stochastic filters (e.g., extended Kalman filters and unscented Kalman filter) in general situations \cite{liu2012state}.

Usually, an intracellular reaction system is modeled by a continuous-time Markov Chain, also known as the stochastic reaction network (SRN), due to  low molecular counts \cite{mcadams1997stochastic}.
Its dynamics can be exactly simulated by the Gillespie stochastic simulation algorithm \cite{gillespie1976general,gillespie1977exact} or the next reaction method \cite{gibson2000efficient}, whose computational complexity is proportional to the rate of the fastest reaction.
Commonly, an SRN has a multiscale nature, meaning that the concentration levels of the involved species differ by several orders of magnitude, and so do the reaction rates.
Consequently, the exact simulation of an intracellular reaction system can take an impractical amount of computational time, and the associated PF, as a simulation-based algorithm, can be computationally inefficient.
Accelerating algorithms for SRNs include the linear noise approximation \cite{van1976expansion}, chemical Langevin equation \cite{gillespie2000chemical} , tau-leaping algorithm \cite{gillespie2001approximate,cao2006efficient,rathinam2003stiffness}, piecewise deterministic Markov process (PDMP) \cite{kang2013separation,hepp2015adaptive,crudu2009hybrid}, etc., all of which
reduce the computational complexity via approximating the firing of fast reactions by more tractable processes. 
(We refer interested readers to the literature \cite{schnoerr2017approximation} for a related discussion.)
Inspired by these methods, researchers also proposed computationally efficient particle filters using linear noise approximations \cite{sherlock2014bayesian},  chemical Langevin equations \cite{liu2012state,golightly2006bayesian}, and PDMPs \cite{fang2020stochastic,fang2021stochastic,hey2015stochastic}.

Another feature of intracellular reaction systems that needs to be considered is the large variability in model parameters from cell to cell, which suggests that model parameters should be inferred simultaneously with the state variables in the filtering problem.
A common way to combine these two tasks together is to view the model parameters as additional states of the system and, then,  to use the PF to infer the augmented state vector.  
In this situation, the conventional sequential importance resampling PF (SIRPF) suffers from sample degeneracy,
meaning that the particle diversity in parameters drops dramatically over time, and, therefore, the particles fail to represent the posterior accurately.
Good alternatives to SIRPFs include the resample-move method \cite{gilks2001following,berzuini2001resample} and the regularized particle filter (RPF) \cite{liu2001combined,oudjane2000progressive}, all of which mitigate sample degeneracy by perturbing particles with (additional) artificial noise and have been successful in various applications (see the aforementioned references).
This methodology is easy to implement but has the drawback that it treats model parameters (which are fixed over time) as time-varying variables and, consequently, throws aways some information about them \cite{liu2001combined}.
In the finite particle population case, researchers still cannot quantify the bias introduced by the artificial noise in an RPF \cite{kantas2015particle}.
However, in the large-particle limit, the RPF has already been shown to converge to the exact filter if the artificial noise is weak and the transition kernel satisfies some regularity conditions, e.g., the globally Lipschitz continuity \cite[Proposition 2.38]{del2000branching}, the mixing condition in the metric space \cite{le2004stability}, and the Lipschitz continuity of transition kernel's derivatives \cite{crisan2014particle}.

Motivated by the above facts, we present in this paper a computationally efficient RPF for multiscale SRNs based on reduced dynamical models and identify conditions under which the performance of the established RPF is guaranteed.
We note that the existing results about the convergence of RPFs cannot be directly applied to our problem. 
First, the globally Lipschitz condition of the transition kernel required in \cite[Proposition 2.38]{del2000branching} is too restrictive for SRNs that only a limited class of networks, such as those having linear propensities or bounded trajectories, satisfy it.
In addition, the discreteness of the state space of an SRN suggests that the mixing condition in the whole metric space \cite{le2004stability} and the derivative related condition \cite{crisan2014particle} can never hold. 
Instead, in this paper, we approach the problem by utilizing parameter sensitivity results in \cite{gupta2013unbiased,gupta2019sensitivity,gupta2018estimation}.
The rationale is that the parameter sensitivity can be used to quantify the effect of the artificial noise exerted to the model parameters; once we show that the error caused by artificial noise is small, the convergence of the RPF holds automatically.
Following this idea, we show that the established RPF converges to the exact filter in probability, if all the reactions that lead to an increase of the molecule population have linearly growing propensities and some other mild conditions are satisfied simultaneously.
The provided conditions are quite general, and they cover a large class of SRNs, including all mass-action networks in which only unimolecular reactions can lead to an increase of the total molecular population.

For the sake of simplicity, we only consider discrete-time observations in this paper. 
However,  we conjecture that the obtained result can be naturally extended to the continuous-time observation case with mild modifications, {as the parameter sensitivity analysis, the major technique used here, applies for continuous observations.} 
Besides, compared with the SIRPF that we provided previously for multiscale SRNs in \cite{fang2021stochastic,fang2020stochastic}, the RPF established in this paper has a relatively similar computational cost but a superior performance in inferring model parameters when the parameter uncertainty is large. 

The rest of this paper is organized as follows.
In \Cref{sec preliminary}, we briefly review the basics of SRNs and introduce the associated filtering problems.
Then, by using different dynamical models, we establish computationally efficient RPFs for SRNs on different scales and provide several mild conditions that guarantee the performance of the established algorithm in the limit of large particles (see \Cref{section RPF}).
When the particle population is finite, we also give a tuning rule for the hyperparameter (the intensity of artificial noise) so that a more precise estimate can be obtained.
In \Cref{Sec. numerical examples}, several numerical examples are presented to illustrate the effectiveness of the constructed RPF, particularly its superior performance to the SIRPF. 
Finally, \Cref{conclusion} concludes this paper.
To improve the readability of the paper, all the proofs are presented in the appendix. 

\section{Preliminary}\label{sec preliminary}

\subsection{Notations}
In this paper, we denote the natural filtered probability space by $( \Omega,  \mathcal F, \{\mathcal F_{t}\}_{t\geq 0}, \mathbb P)$, where $\Omega$ is the sample space, $\mathcal F$ is the $\sigma$-algebra, $\{\mathcal F_{t}\}_{t\geq 0}$ is the filtration, and $\mathbb P$ is the natural probability.
Also, we term $\mathbb N_{>0}$ as the set of positive integers, $\mathbb R^{n}$ with $n$ being a positive integer as the space of $n$-dimensional real vectors, $\|\cdot\|_2$ as the Euclidean norm, and $|\cdot|$ as the absolute value notation. 
For any positive integer $n$ and any $t>0$, we term $\mathbb D_{ \mathbb R^{ n}}[0,t]$ as the Skorokhod space that consists of all $\mathbb R^n$ valued cadlag functions on $[0,t]$. 

\subsection{Stochastic reaction networks}
In this paper, we consider an intracellular reaction system that consists of $n$ chemical species ($S_1, \dots ,S_n$)
and $r$ reactions:
\begin{equation*}
v_{1,j} S_1 + \dots + v_{n,j} S_n \ce{ ->} v'_{1,j} S_1 + \dots + v'_{n,j} S_n, \quad j=1,\dots, r,
\end{equation*}
where $v_{i,j}$ and $v_{i,j}$ are stoichiometric coefficients representing numbers of molecules consumed and produced in each reaction.
We term $X(t)=\left(X_1(t),X_2(t), \dots, X_n(t)\right)^{\top}$ as the vector of copy numbers of these species at time $t$.
Due to the small number of reactant molecules, an intracellular reaction system is usually model by a stochastic model, called the continuous-time Markov chain, whose dynamics are described by \cite{anderson2011continuous}
\begin{equation*}
X(t)= X(0)+\sum_{j=1}^{r} \zeta_j R_j\left( \int_0^t \lambda_j (k, X(s))\dd s \right)
\end{equation*}
where $\zeta_j\triangleq v'_{\cdot j}- v_{\cdot j}$ (for $j=1,\dots,r$),  $R_{j}(t)$ are independent unit rate Poisson processes, $k= (k_1,\dots,k_{\tilde r})^{\top}$ is an $\tilde r$-dimensional vector of model parameters, and $\lambda_j (\kappa,x)$ are propensities indicating the speed of reaction firing.
The most common propensity in chemistry and biology is the mass-action type, where $\lambda_j (\kappa,x)= k_j \prod_{i=1}^n \frac{x_i!}{\left(x_i-v_{i,j}\right)!}\mathbbold 1_{\{x_i\geq v_{i,j}\}}$ with $\mathbbold{1}_{\{\cdot\}}$ the indicator function.
In the literature, the above model describing the dynamics of intracellular reaction systems is termed as a stochastic reaction network (SRN) due to the stochasticity of the dynamics and the network structure of the chemical reactions.

Usually, an SRN encountered in systems biology is multiscale in nature \cite{kang2013separation}, meaning that different species vary a lot in abundance, and so do reaction rates. 
Following \cite{kang2013separation}, these quantities at different scales can be normalized as follows.
Let $N$ be a scaling factor, the parameter $\alpha_i$  the abundance factor of the $i$-th species such that $X^{N}_i(t)\triangleq N^{-\alpha_i}X_{i}(t) = O(1)$, and the parameter $\beta_j$ the magnitude of $k_{j}$ such that $k'_j\triangleq k_j N^{-\beta_j}=O(1)$.
Similarly, we term $\rho_j$ as the timescale of the $j$-th reaction such that $\lambda^{N}_j (\mathcal K, X^N(t))\triangleq N^{-\rho_j}\lambda_j (k, X(t))= O(1)$ where $\mathcal K\triangleq \left(k'_1,\dots,k'_{\tilde r}\right)^{\top}$.
Particularly, in the mass-action kinetics, these coefficients satisfy $\rho_j=\beta_j + \sum_{i=1}^{n}v_{ij} \alpha_i$.
Finally, by denoting $X^{N,\gamma}(t) \triangleq X^{N}(tN^{\gamma})$ with $\gamma$ the timescale of interest, we can re-express the dynamics as
\begin{align}{\label{eq. scaling stochastic dynamics}}
X^{N,\gamma}(t) =& X^{N,\gamma}(0)   +\sum_{j=1}^{r} \Lambda^{N} \zeta_j R_j\left( \int_0^{t} N^{\gamma+ \rho_j} \lambda^N_{j}(\mathcal K, X^{N,\gamma}(s)) \dd s \right) 
\end{align}
where $\Lambda^{N}\triangleq \text{diag}(N^{-\alpha_1},\dots,N^{-\alpha_s})$, and all terms but $\Lambda^{N}$ and $N^{\gamma+\rho_j}$ are in the order of $O(1)$.
In the sequel, we study the SRN in the normalized coordinate.
Also, for simplicity, we concern ourselves with the fastest timescale of the system, i.e., $\gamma_1\triangleq \min_{i}\left\{\alpha_i-\max_{j\in \{v_{i,j}\neq 0\}} \rho_j \right\}$, on which the system evolves at the rate of $\mathcal O (1)$.

\subsection{Filtering problems for SRNs}{\label{subsection filtering problems}}
Stochastic filtering for SRNs aims to infer hidden dynamic states of an SRN from partial observations of the system.
Usually, we need to simultaneously estimate the model parameter $\mathcal K$ in this filtering problem due to the uncertainty of $\mathcal K$.
In this paper, we suppose $\left(\mathcal K,X^{N}(0)\right)$ to have  a known prior distribution. 

For the observation, we assume that $m$ channels of light intensity signals, denoted by $Y^{N,\gamma_1}(\cdot)$, can be observed (with a microscope for instance) and satisfy
\begin{align*}
Y^{N,\gamma_1}(t_i)&= h\left(X^{N,\gamma_1}(t_i)\right)  + W(t_i) \qquad && \forall i\in \mathbb N_{>0},
\end{align*}
where $\{t_i\}_{i\in\mathbb N_{>0}}$ is a strictly increasing sequence of time points at which the observation is collected, $h$ is an $m$-dimensional bounded Lipschitz continuous function indicating the relation between the observation and the SRN, and $\{W_{\ell}(t_i)\}_{i\in\mathbb N_{>0}}$ are independent $m$-variate standard Gaussian random variables.
Usually, the dimension of the observation vector is much less than that of the state vector, meaning that only a part of the system information is detected.  
Besides, we assume that the observation noise is independent of the dynamical system \eqref{eq. scaling stochastic dynamics}, i.e., $\{W_{\ell}(t_i)\}_{i\in\mathbb N_{>0}}$ are independent of Poisson processes $R_{j} (\cdot)$ ($j=1,\dots,r$), system parameters $\mathcal K$, and the initial condition $X^{N}(0)$.
In the sequel, we term $\mathcal Y^{N,\gamma_1}_{t_i}$ as the $\sigma$-algebra generated by observations up to the time $t_i$.

In the filtering problem, the standard object to study is the conditional distribution process $\{\pi^{N,\gamma_1}_{t_i}\}_{i\in\mathbb N_{>0}}$ given by  $\pi^{N,\gamma_1}_{t_i}(\phi)\triangleq \mathbb E_{\mathbb P}\left[\phi\left(\mathcal K, X^{N,\gamma_1}(t_i)\right) \left | \mathcal Y^{N,\gamma_1}_{t_i} \right.\right]$ for any measurable function $\phi:\mathbb R^{\tilde r} \times \mathbb R^{n} \to \mathbb R$.
This conditional distribution process satisfies the following relations \cite[Proposition 10.6]{bain2008fundamentals}
\begin{align}
p^{N,\gamma_1}_{t_i}(A) &= \int_{z\in \mathbb R^{\tilde r} \times \mathbb R^n } K^{N,\gamma_1}_{t_i-t_{i-1}} (z, A) \pi^{N,\gamma_1}_{t_{i-1}}(\dd z), && \forall A \subset \mathbb R^{\tilde r} \times \mathbb R^n \text{ and } \forall i \in\mathbb N_{>0}, \label{eq. recurssive expression 1}\\
\frac{\dd \pi^{N,\gamma_1}_{t_i}}{ \dd p^{N,\gamma_1}_{t_i} } &= \frac{ \ell_{Y^{N,\gamma_1}(t_i)} }{p^{N,\gamma_1}_{t_i} \left(\ell_{Y^{N,\gamma_1}(t_i)}\right) }, && \forall i \in\mathbb N_{>0}, \label{eq. recurssive expression 2}
\end{align}
where $t_0=0$, $\pi^{N,\gamma_1}_{t_0}$ is the initial distribution, $K^{N,\gamma_1}_{t}(z, A)$ is the transition kernel  of the system \eqref{eq. scaling stochastic dynamics} being equal to $ \mathbb P\left( (\mathcal K, X^{N,\gamma_1}(t))\in A \left| (\mathcal K, X^{N,\gamma_1}(0))=z \right.\right)$, and $\ell_{y}(\kappa, x)\triangleq
\mathbb P \left( Y^{N,\gamma_1}(t_i)=y | X^{N,\gamma_1}(t_1)=x \right)$ is the likelihood function. 
Here, equation \eqref{eq. recurssive expression 1} can be interpreted as a prediction step where one estimates hidden dynamic states at $t_i$ using observations up to the previous time points, and \eqref{eq. recurssive expression 2} is an adjustment step where the conditional distribution is obtained by reweighing the prediction according to the new observation. 
Usually, the above equations cannot be solved explicitly; therefore, we intend to establish efficient particle filtering algorithms that can quickly and accurately solve these filtering problems.

\subsection{Model reduction for SRNs}

The dynamics of an SRN is a continuous-time Markov chain (CTMC) which can be exactly simulated using the Gillespie stochastic simulation algorithm \cite{gillespie1976general,gillespie1977exact} or the next reaction method \cite{gibson2000efficient}.
The computational complexity of these exact simulation algorithms is inversely proportional to the fastest timescale of reactions;
therefore, these algorithms can only be computationally efficient for those SRNs having no fast reaction (i.e., the maximum $N^{\gamma_1+\rho_j}$ is small), otherwise they take an impractical amount of computational time to simulate systems.
This fact can also have an adverse effect on particle filtering, as it is a simulation-based method (more details are illustrated in the next section).
To accelerate the simulation, a few algorithms have been invented for systems on different scales. 
In the following, we review two well-known reduced models, the tau-leaping model and the piecewise deterministic Markov process (PDMP).

\subsubsection{Tau-leaping}
The tau-leaping method is a stochastic analog to the Euler method for ODEs, where the algorithm changes the value of the propensities periodically over time rather than instantaneously after each jump event. 
Since the propensity is updated much less frequently in this method than in the exact simulation method, the tau-leaping algorithm consumes much less computational time to simulate the system. 
At the timescale $\gamma_1$, the tau-leaping algorithm can be expressed as
\begin{align}
X^{N,\gamma_1}_{\tau}(t)= X^{N}(0)
+\sum_{j=1}^{r} \Lambda^{N} \zeta_j R_j\left( \sum_{i=0}^{\infty} 
\left({\tau_{i+1}\wedge t}- {\tau_{i}\wedge t} \right) N^{\gamma_1+ \rho_j} \lambda^N_{j}\left(\mathcal K, X^{N,\gamma_1}_{\tau}\left( \tau_i \right)\right)  \right)
\end{align}
where $X^{N,\gamma_1}_{\tau}(t)$ is the state vector, and $\{\tau_i\}_{i\in\mathbb N}$ is an increasing sequence of time points at which propensities are updated. Equivalently, it can also be written by $$X^{N,\gamma_1}_{\tau}\left(\tau_{i+1}\right)= X^{N,\gamma_1}_{\tau}\left(\tau_i\right)
+\sum_{j=1}^{r} \Lambda^{N} \zeta_j \mathcal P_{j,i}$$
where $\mathcal P_{j,i}$ are independent Poisson-distributed random variables with mean $(\tau_{i+1}-\tau_i) N^{\gamma_1+ \rho_j} \lambda^N_{j}\left(\mathcal K, X^{N,\gamma_1}_{\tau}(\tau_i)\right)$.
For simplicity, in this paper, we consider the case where the grid $\{\tau_i\}_{i\in\mathbb N}$ is deterministic and contains all the observation times.
Usually, the leaping time should be short enough so that the normalized propensity will not change much during $[\tau_i, \tau_{i+1})$.
Also, it should be greater than the timescale of the fastest reaction so that it leaps over several reaction firing events and, thus, is faster than the exact simulation method. 
Often, the tau-leaping algorithm is applied to systems containing $ 10^{2}\sim 10^3$  molecules \cite{hu2011weak} (in other words, the maximum $N^{\gamma_1+ \rho_j}$ is $\mathcal O \left(10^{2}\sim 10^3\right)$).

In this paper, we only consider applying the tau-leaping algorithm to the system where the following assumptions hold.
\begin{assumption}[Nonnegativity of the tau-leaping algorithm]{\label{assumption non-negativity of tau-leaping algorithm}}
	Denote $|\tau|\triangleq \max_{i} (\tau_{i+1}-\tau_i)$. Then, for any $t>0$ and any $|\tau|>0$, the state $X^{N,\gamma_1}_{\tau}(t)$ is always nonnegative.
\end{assumption}
\begin{assumption}[Consistency of the tau-leaping algorithm]{\label{assumption convergence of tau-leaping algorithm}}
	For any fixed $\mathcal K$, there holds $X^{N,\gamma_1}_{\tau}(\cdot) \stackrel{|\tau| \to 0}{\Longrightarrow} X^{N,\gamma_1}(\cdot) $ on any interval $T>0$, in the sense of the Skorokhod topology.
\end{assumption}
In general, the above two properties assumed for the tau-leaping method are essential and non-restrictive. 
The nonnegativity assumption ensures the algorithm to generate reasonable trajectories that evolve in the nonnegative orthant, and this property can be achieved by wisely designing the grid $\{\tau_i\}_{i\in\mathbb N}$ (see \cite{anderson2008incorporating}).
The consistency assumption states that the tau-leaping algorithm is accurate (in the weak sense) in simulating the target system when the coarseness of the time-discretization scheme is small enough. 
A few sufficient conditions for the consistency assumption have been provided in the literature \cite{anderson2011error,li2007analysis,rathinam2005consistency} (for bounded processes) and \cite{rathinam2016convergence} (for unbounded processes).

\subsubsection{Piecewise deterministic Markov Processes (PDMP)}
For a system having an extremely large scaling factor $N$, the PDMP is a favorable approximate model, where the firing events of fast reactions are approximated by continuous processes, and those slow reactions are kept as what they are. 
Specifically, for fast reactions ($j\in \{\gamma_1+ \rho_j>0\}$), this method utilizes Poisson's law of large numbers to approximate the reaction firing events by continuous processes, i.e.,
\begin{equation*}
N^{-\gamma_1- \rho_j} \zeta_j R_j\left( \int_{0}^{t} N^{\gamma_1+ \rho_j} \lambda^N_{j}\left(\mathcal K, X^{N,\gamma_1}\left( s \right)\right) \dd s \right) 
\approx D^{\gamma_1+ \rho_j}\zeta_j \left(\int_0^t \lambda'_j( \mathcal K, X^{\gamma_1}(s))  \dd s\right)
\end{equation*}
where $\lambda'_j(\kappa, x)=\lim_{N\to+\infty} \lambda^N_j (\kappa, x)$, $X^{\gamma_1}(\cdot)$ is the limit process of $X^{N,\gamma_1}(\cdot)$, and $D^{\tilde \alpha} \triangleq \text{diag}( \mathbbold 1_{\{\alpha_1=\tilde \alpha\}}, \dots,  \mathbbold 1_{\{\alpha_n=\tilde \alpha\}})$ indicates whether a species has the abundance factor $\tilde \alpha$.
In this regard, the reduced model can be written by \cite{kang2013separation}
\begin{align}{\label{eq. reduced model at the first time scale}}
X^{\gamma_{1}}(t)=& \lim_{N\to \infty} X^{N}(0) + \sum_{j : \gamma_1+ \rho_j>0} D^{\gamma_1+ \rho_j}\zeta_j \left(\int_0^t \lambda'_j( \mathcal K, X^{\gamma_1}(s))  \dd s \right)   \\
& +  \sum_{j : \gamma_1+ \rho_j=0 } D^0\zeta_j 
R_{j}\left(\int_0^t \lambda'_j( \mathcal K, X^{\gamma_1}(s)) \dd s \right)  \notag
\end{align}
where $\lim_{N\to \infty} X^{N}(0)$ is assumed to exist.
Since the above Markov process is deterministic between two neighboring jumping times, it is named the piecewise deterministic Markov process (PDMP).
For a slower timescale, such a reduced model can also be constructed in the same spirit if some fast species are first eliminated by quasi-stationary assumption. 
A systematic procedure to construct a hierarchy of PDMPs representing dynamics in different timescales can be found in the literature \cite{kang2013separation}.

When $N$ is large, the simulation time of PDMPs is much less than that of the full model, because the PDMP avoids the exact simulation of fast reactions, and those continuous processes that approximate the firing events of fast reactions can be computed efficiently by the Euler method.
Detailed algorithms to simulate PDMPs have been provided in the literature \cite{crudu2009hybrid,duncan2016hybrid,hepp2015adaptive}.
The accuracy of the PDMP is ensured by the following proposition. 
\begin{assumption}\label{assumption convergence of the initial condition}
	$\lim_{N\to \infty} X^{N}(0)$ exists a.s., and $\lim_{N\to \infty} X^{N}_i(0)>0$ if  $\alpha_i>0$.
\end{assumption}
\begin{proposition}[Adapted from \cite{kang2013separation}]{\label{prop Kang kurtz gamma1}}
	For any fixed $\mathcal K$, if \cref{assumption convergence of the initial condition} holds, and $X^{N,\gamma_1}(t)$ (for all $N\in\mathbb N_{>0}$) and $X^{\gamma_1}(t)$ are almost surely non-explosive, then $X^{N,\gamma_1}(\cdot ) \stackrel{N\to\infty}{\Longrightarrow} X^{\gamma_{1}}(\cdot)$ in the sense of the Skorokhod topology on any finite time interval $[0,T]$.
\end{proposition}

\subsubsection{Working ranges of different simulation algorithms}{\label{subsubsection tau-leaping and PDMP}}
	The working ranges of the aforementioned simulation algorithms are quite different. 
	When all reactions fire at the rate of $\mathcal O(1\sim 10)$ (i.e., all $N^{\gamma_1+ \rho_j}$ are $\mathcal O(1\sim 10)$), the exact simulation algorithms for the CTMC are recommended, as they are perfectly accurate and computationally efficient in this case.
	The tau-leaping algorithm is usually applied to systems where the fastest reaction fires at the rate of $\mathcal O \left(10^{2}\sim 10^3\right)$;
	in this scenario, the tau-leaping algorithm can be much faster than the exact simulation method while still capturing key behaviors of the system. 
	Finally, when the scaling factor $N$ is sufficiently large, the PDMP method is preferred to the others, because in this situation the randomness of fast reaction firing is averaged out due to the law of large numbers. 
	
	There is still one situation excluded here, that some reactions fire at extremely fast rates (e.g., $\mathcal O \left(10^{4}\right)$), but some others fire at the rate of $\mathcal O \left(10^{2}\sim 10^3\right)$.
	In this case, the PDMP model fails to capture key behaviors of the system, as there is still randomness in those reactions on medium timescales.
	On the other hand, the tau-leaping method is slow in this situation, because it needs to frequently  generate Poisson random variables of large means.
	Instead, the diffusion approximation is recommended for these systems, where the firing of a medium timescale reaction is approximated by a diffusion process.
	Since this algorithm just needs to generate Gaussian random variables, the computational cost is low;
	moreover, its accuracy is higher than that of the PDMP \cite{kang2014central}.
	However, in this paper, we do not consider this type of approximation due to the difficulties in the analysis (particularly the parameter sensitivity analysis of this reduced model).

\section{Regularized particle filter (RPF) for SRNs} \label{section RPF}

\subsection{Constructing RPFs for SRNs on different scales}
Usually, the conditional distribution of an SRN cannot be computed explicitly due to the nonlinearity of the dynamics.
Therefore, we apply particle filters to the associated filtering problem. 
The key idea of the particle filtering (also known as the sequential Monte Carlo method \cite{gordon1993novel}) is to approximate the posterior by a number of weighted trajectories sampled from the distribution of the target system.
A detailed particle filtering algorithm is provided in \Cref{alg discrete time regularized particle filters}.
Specifically,  the sampling step in \Cref{alg discrete time regularized particle filters} mimics the recursive formulas \eqref{eq. recurssive expression 1} (the prediction step) and \eqref{eq. recurssive expression 2} (the adjustment step), in which weighted particles are generated to represent the conditional distribution.
The resampling step is executed to remove insignificant particles and, therefore, mitigate long-term weight degeneracy at the cost of adding additional noise at the current step \cite{doucet2009tutorial}. 
In this paper, we use the residual resampling, as it introduces the minimum noise among all resampling methods \cite[Exerice 9.1]{bain2008fundamentals}.
Finally, particles are perturbed by artificial noise so that the particle diversity improves, and the sample degeneracy meaning that all particles become identical can be avoided.
Such a particle filtering algorithm with an artificial evolution step is called a regularized particle filter (RPF); if the artificial evolution is not executed, then this algorithm is called a sequential importance resampling particle filter (SIRPF).

\begin{algorithm}[htbp]
	\caption{Regularized particle filter \cite{oudjane2000progressive}}
	\label{alg discrete time regularized particle filters}
	
	\begin{algorithmic}[1]
		\STATE  Input observations $\{Y(t_i)\}_{i\in\mathbb N}$, a dynamical model, and the initial distribution.
		\STATE  \textit{Initialization:} Sample $M$ particles $(\bar \kappa_1(0), \bar x_1(0)),\dots, (\bar \kappa_M(0), \bar x_M(0))$; set $i=1$ and $t_0=0$.
		\WHILE{$t_i$ does not exceed the terminal time of the observations}
		\STATE  \textit{Sampling:} simulate $x_j(\cdot)$ ($j=1,\dots,M$) from time $t_{i-1}$ to $t_i$ according to the dynamical model with parameters $\bar \kappa_j(t_{i-1}) $ and initial conditions $\bar x_j (t_{i-1})$;  Set $\kappa_j (t_{i})=\bar \kappa_j (t_{i-1})$; Calculate weights $w_j(t_i)\propto \ell_{Y^{N,\gamma}(t_i)}(\kappa_j(t_i), x_j(t_i))$.
		\STATE  \textit{Output the filter:}  $\bar \pi_{M,t_i} (\phi)  = {\sum_{j=1}^{M} w_j(t_i) \phi(\kappa_j, x_j(t_i))}$.
		\STATE  \textit{Resampling:} Resample 
		$\{ w_{j}(t_i), (\kappa_j(t_i), x_j(t_i))\}_j$ to obtain M equally weighted samples $\{{1}/{M}, (\hat \kappa_j(t_i), \bar x_j(t_i))\}_j$. (Residual resampling is applied)
		\STATE  \textit{Artificial evolution:} sample $\bar \kappa_j(t_i)$ from a kernel $\eta_{M}\left(\cdot | \hat \kappa_j(t_i) \right)$. 
		\ENDWHILE
	\end{algorithmic}
\end{algorithm}

Conventionally, in the sampling step, particles are obtained by simulating the full dynamical model of the system.
However, it is not plausible for SRNs, especially the multiscale ones, where the exact simulation of \eqref{eq. scaling stochastic dynamics} can take an impractical amount of computational time. 
Instead, for SRNs on different scales, we utilize different dynamical models in the sampling step so that the filter is computationally efficient. 
The specific construction of such RPFs is as follows.
\begin{definition}[RPFs for SRNs on different scales] \label{def RPFs for SCRNs}
\begin{itemize}
	\item If all reactions fire at the rate of $\mathcal O(1\sim 10)$ (i.e., all $N^{\gamma_1+ \rho_j}$ are $\mathcal O(1\sim 10)$), then we build the filter by utilizing exact simulation algorithms in the sampling step. 
	We denote such a RPF by $\bar \pi^{N,\gamma_1}_{M,t_i}$.
	\item If the fastest reaction fires at the rate of $\mathcal O \left(10^{2}\sim 10^3\right)$, then we use the tau-leaping method in the sampling step. 
	We denote such a RPF by $\bar \pi^{N,\gamma_1,\tau}_{M,t_i}$.
	\item If $N$ is extremely large, we use the PDMP in the sampling step. We denote such a RPF by $\bar \pi^{N,\gamma_1, H}_{M,t_i}$. (Here, $H$ stands for hybrid approximations. This filter has dependence on $N$ because the observation $Y^{N,\gamma_1}(\cdot)$ is input into the algorithm.)
\end{itemize}
\end{definition}

Another important factor for the RPF is the artificial noise used in the artificial evolution step. 
Though this noise can mitigate sample degeneracy, it also introduces a bias to the estimate.
Therefore, if the kernel $\eta_{M}$ is not carefully chosen, the RPF can yield a very inaccurate estimate of the conditional distribution.
In this paper, we design the kernel $\eta_{M}$ according to the following assumptions so that the intensity of the artificial noise is contained.
\begin{assumption}{\label{cond. conditions on the artificial noise 1}}~
	\begin{itemize}
		\item $\mathcal K$'s prior distribution has a compact convex support. We denote it by $\Theta$.
		\item We term $\mathcal V$ as the covariance matrix of $\mathcal K$ at the initial time. Then, for any $M>0$ and $\kappa\in\Theta$, the kernel $\eta_{M}\left(\cdot | \kappa\right)$ satisfies
		\begin{equation*}
		\eta_{M}\left( \tilde \kappa | \kappa\right) \propto \mathbbold 1_{\{\tilde \kappa\in \Theta \}} \mathcal N\left(\tilde \kappa \left| \kappa, C^2_{\eta} \mathcal V/M \right. \right) 
		\qquad \forall \tilde \kappa\in\mathbb R^{\tilde r},
		\end{equation*}
		where $C_{\eta}$ is a predetermined positive constant, and $\mathcal N\left(\cdot \left| \kappa, C^2_{\eta} \mathcal V/M \right. \right)$ is a Gaussian kernel with $\kappa$ mean and $C^2_{\eta} \mathcal V/M$ covariance matrix.
	\end{itemize}
\end{assumption}
In general, the above assumption suggests that the artificial evolution keeps particles in the support of the prior distribution, and the artificial noise is in the order of $O\left( M^{-1/2}\right)$.

\begin{remark}{\label{remark artificial evolution}}
	In references \cite{liu2001combined,liu2012state}, there is an alternative choice of the artificial noise where $\eta_{M}\left( \tilde \kappa | \kappa\right)=\mathcal N\left(\tilde \kappa \left| a\kappa+(1-a) \mu_t, (1-a) C^2_{\eta} \mathcal V_t/M \right. \right)$  with $\mu_t$ and $\mathcal V_t$ the mean and covariance matrix of particles $\{\hat k_j(t_i)\}_{1\leq j\leq M}$ and `$a$' a constant in $[0,1]$.
	Compared with our setting, this one can adjust the artificial noise intensity dynamically according to the particle distribution, and, moreover, the associated filter can avoid ``variance inflation" \cite{liu2001combined}, meaning that the artificial evolution in its filter does not change the particle variance.
	However, we think that variance inflation is beneficial, and with the method in \cite{liu2001combined,liu2012state}, the RPF can only delay sample degeneracy but never fully solve it.
    The reasons are as follows.
    With more observations collected, the exact posterior of the model parameter becomes narrower, and the variance of the particles corresponds to $\mathcal K$ inclines to decay dramatically over time.
	Since the method in \cite{liu2001combined,liu2012state} avoids variance inflation, its particle variance decays strictly over time and finally degenerates to zero; in other words, the particles $\{ k_j(t_i)\}_{1\leq j\leq M}$ collapse into a single point in the end.
	Ideally, the point should be close to the maximum of the posterior.
	However, due to the randomness, this ideal case can rarely happen, and, consequently, sample degeneracy arises. 
    In contrast, the RPF satisfying \cref{cond. conditions on the artificial noise 1} 
    can greatly improve the particle variance (when it is too small) due to the constant covariance matrix of the artificial noise and, therefore, can fully solve the sample degeneracy problem.

	
	Another fact to support our choice of artificial noise is that in reality the model parameters fluctuate dynamically in individual living cells due to environmental changes and the cell cycle \cite{rosenfeld2005gene}.
	This requires the algorithm to generate diverse particles even if the exact posterior is very narrow.
	Again, due to the constant covariance matrix of the artificial noise, our framework meets this requirement. 
\end{remark}

\begin{remark}
	Compared with our previous work \cite{fang2021stochastic,fang2020stochastic}, where an SIRPF was established for SRNs based on the PDMP, this paper constructs filters using one more reduced model, the tau-leaping method. 
	More importantly, the RPF constructed in this paper involves an artificial evolution step, which can mitigate sample degeneracy and make the filter perform better when the parameter uncertainty is large. 
	In the next section, we will illustrate this point using several numerical examples. 
\end{remark}

\subsection{Convergence of RPFs}
In this subsection, we investigate conditions under which the established RPFs can provide reliable estimates of the hidden dynamic states and model parameters in the limit of large particles.
In the literature \cite{del2000branching,le2004stability,crisan2014particle}, a few conditions have been proposed to guarantee the convergence of RPFs;  however, none of them can be directly applied to our setting for the following reasons.

\begin{remark}
	It has been shown that an RPF converges to the exact filter if the artificial noise is weak and the transition kernel satisfies some regularity conditions, e.g., the globally Lipschitz continuity \cite[Proposition 2.38]{del2000branching}, the mixing condition in the metric space \cite{le2004stability}, and the Lipschitz continuity of transition kernel's derivatives \cite{crisan2014particle}.
    For an SRN, the globally Lipschitz continuity of the transition kernel is very restrictive; it usually requires the propensities to be Lipschitz continuous (e.g., linear propensities), or the trajectory to evolve in a bounded region.  
    In addition, the discrete state space of an SRN invalidates the mixing condition in the whole metric space \cite{le2004stability} and the derivative related condition in \cite{crisan2014particle}. 
\end{remark}

In this paper, we approach this problem in an alternative way. 
Note that when the intensity of the artificial noise is zero, the constructed RPF degenerates to an SIRPF, and its convergence has been shown in part in our previous work \cite{fang2021stochastic}.
Therefore, the key to proving the convergence of the established RPFs lies in estimating the bias introduced by the artificial evolution and showing it to vanish as particle population tends to infinity. 
This corresponds to calculating the parameter sensitivity for SRNs, as the artificial noise only perturbs the model parameter part of the particles.
Fortunately, neat expressions for this parameter sensitivity function have been provided in the literature \cite{gupta2013unbiased,gupta2019sensitivity,gupta2018estimation}, by which we can easily analyze the bias of interest. 
Finally, through some technical analyses, we can show that both the parameter sensitivity and the state vector grow mildly under some weak conditions, and, consequently, the constructed RPF is convergent in probability to the exact filter.
More detailed discussions about it are presented in the Appendix.

We now specify a few assumptions needed in our main results. 
In \Cref{remark covered SCRNs}, we illustrate that such conditions are so mild that they cover a large class of SRNs. 

\begin{assumption}{\label{assumption for full model}}
	~
	\begin{itemize}
		\item All the propensities $\lambda^N_j (\kappa, x)$ are twice differentiable with respect to $\kappa$. Moreover, there exist positive constants $C_{\lambda,1}$ and $q$ such that
		\begin{equation*}
		\max\left\{
		\left|\lambda^N_j (\kappa, x)\right|,
		\left\|\frac{\partial \lambda^N_j (\kappa, x)}{\partial \kappa}\right\|_2,
		\left\|\frac{\partial^2 \lambda^N_j (\kappa, x)}{\partial \kappa^2}\right\|_2
		\right\} \leq C_{\lambda,1} \left(1+\left\|x\right\|_2^q\right)
		\end{equation*} 
		for any $N>0$, $j\in\{1,\dots,r\}$, $\kappa\in \Theta$, and $x\in \mathbb{R}^{n}_{\geq 0}$.
		\item  For any $N>0$ and $j\in\{1,\dots, r\}$, if $\lambda^N_j (\kappa, x)>0$, then $x+\Lambda^{N} \zeta_j $ has no negative component. 
		\item  For any $p>0$, there exists a constant $ C_{0,p}$ such that  $\mathbb E \left[  \left\| X^{N}(0)\right\|_2^p \right]\leq C_{0,p}$
		\item  There exists a positive constant $C_{\lambda,2}$ such that for any $N$ and $j$ satisfying $ \bar 1^{\top} \Lambda^{N}  \zeta_j >0$ ($\bar 1$ is the all-one vector) there holds
		\begin{equation*}
		\lambda^N_j (\kappa, x) \leq C_{\lambda,2}\left(1+\|x\|_2\right)\qquad \forall (\kappa,x)\in \Theta\times \mathbb R^{n}_{\geq 0}.
		\end{equation*}
		In other words, the propensity of the reaction that leads to an increase of the total population has at most a linear growth rate with the state argument. 
	\end{itemize}
\end{assumption}

\begin{assumption}{\label{assumption for PDMPs}}
For any reaction $j$ and parameter $\kappa\in\Theta$, the convergence $\lambda^N_j(\kappa, \cdot) \to \lambda'_j(\kappa, \cdot) $ (as $N\to\infty$) holds uniformly over compact sets. The same is true for derivatives $\frac{\partial \lambda^N_j (\kappa, \cdot)}{\partial \kappa} \to \frac{\partial \lambda'_j (\kappa, \cdot)}{\partial \kappa}$ and $\frac{\partial \lambda^N_j (\kappa, \cdot)}{\partial x} \to \frac{\partial \lambda'_j (\kappa, \cdot)}{\partial x}$.
\end{assumption}

Then, we present the main result of this paper, which states that the constructed RPFs can be very close to the exact filter when the particle population is large.

\begin{theorem}{\label{thm main}}
	Assume that \Cref{cond. conditions on the artificial noise 1} and \Cref{assumption for full model} hold.
	Also, we suppose the test function $\phi:\Theta\times \mathbb R^{n}_{\geq 0} \to \mathbb R$ to be bounded, continuous, and continuously differentiable in the first $m$-coordinates and satisfy
	\begin{equation}\label{eq. sensitivity condition of phi}
	\left\|\frac{\partial  \phi (\kappa, x)}{\partial \kappa}\right\|_2 \leq C_{\phi} \left(1+\left\|x\right\|_2^q\right)
	\qquad \forall (\kappa,x)\in \Theta \times \mathbb R^{n}_{\geq 0},
	\end{equation}
	where $C_{\phi}$ is a predetermined constant, and $q$ is the same as that in \Cref{assumption for full model}.
	Then, the following results for the constructed RPFs hold for any  $i\in\mathbb N_{>0}$.
	\begin{itemize}
		\item 
		$ \lim_{M\to\infty}\mathbb E \left[\left|\bar \pi^{N,\gamma_1}_{M,t_i}(\phi)- \pi^{N,\gamma_1}_{t_i}(\phi) \right| ^2 \right] =0$ 
		for any scaling factor $N$, where $\bar \pi^{N,\gamma_1}_{M,t_i}(\phi)$ is the filter built on exact simulation algorithms, and $\pi^{N,\gamma_1}_{t_i}(\phi)$ is the exact filter.
		\item If, moreover, \Cref{assumption non-negativity of tau-leaping algorithm} and \Cref{assumption convergence of tau-leaping algorithm} hold, then there is the relation 
		$ \lim_{M\to\infty} \lim_{|\tau|\to 0} \mathbb E \left[ \left|\bar \pi^{N,\gamma_1,\tau}_{M,t_i}(\phi)- \pi^{N,\gamma_1}_{t_i}(\phi) \right|^2 \right] =0$ 
		for any $N$, where $\bar \pi^{N,\gamma_1,\tau}_{M,t_i}(\phi)$ is the filter built on the tau-leaping method.
		\item If \Cref{assumption convergence of the initial condition} and \Cref{assumption for PDMPs} also hold, then there is the relation 
		$ \lim_{M\to\infty} \lim_{N\to +\infty} \mathbb E \left[\left|\bar \pi^{N,\gamma_1,H}_{M,t_i}(\phi)- \pi^{N,\gamma_1}_{t_i}(\phi) \right|^2\right] =0$,
		where $\bar \pi^{N,\gamma_1,H}_{M,t_i}(\phi)$ is the filter built on the PDMP model.
	\end{itemize}
\end{theorem}
\begin{proof}
	The proof is shown in the {appendix}.
\end{proof}

\begin{remark}{\label{remark covered SCRNs}}
	The conditions required in the above theorem are not restrictive.
	\Cref{assumption non-negativity of tau-leaping algorithm}, \Cref{assumption convergence of tau-leaping algorithm}, and \Cref{cond. conditions on the artificial noise 1} can be easily achieved by properly designing the tau-leaping algorithm and the RPF.
	\Cref{assumption convergence of the initial condition} and \Cref{assumption for PDMPs} are concerned with the well-definiteness of the PDMP, which are reasonable to assume for those systems where the PDMP has already been successfully applied for simulation.
	For \Cref{assumption for full model}, the first claim (the polynomial growth rate of the propensity) is satisfied in most SRNs considered in the existing literature (following the mass-action kinetics or the Hill-type dynamics); the second claim is nothing but ensuring the state vector to be nonnegative;
	the third requires all moments of initial conditions to be finite, which is natural in biology.
	The last claim in \Cref{assumption for full model} indicates that all reactions that lead to an increase of the molecular population have linearly growing propensities.
	Though the last one is not as weak as other assumptions, it still includes a lot of SRNs in systems biology, e.g., the simple birth-death model, the antithetic integral controller that interconnects to linear networks  (see \cite{briat2016antithetic,aoki2019universal} for antithetic integral controllers), and all mass-action networks  where only unimolecular reactions can lead to an increase of the total molecular population.
	In summary, the conditions required in \Cref{thm main} are quite mild, and the proposed RPFs work for a large class of SRNs.
\end{remark}

\subsection{Hyperparameter tuning}
In the finite particle population case, the hyperparameter $C_{\eta}$, which depicts the intensity of the noise in artificial evolution  (see \Cref{cond. conditions on the artificial noise 1}), has a significant influence on the performance of the established filter and, therefore, needs to be properly trained in advance.
Intuitively, this hyperparameter tuning task is an exploration-exploitation trade-off: by increasing the hyperparameter $C_\eta$, the RPF can explore a broad region of the parameter space and have a higher chance to discover all regions of high posterior, but at the cost of missing some knowledge conveyed in the current particles. 

To optimize the hyperparameter, we generate a large number of simulated trajectories, $\{(\tilde \kappa_{\ell}, \tilde x_\ell(\cdot),\tilde y_\ell(\cdot))\}$, where $(\tilde \kappa_{\ell}, \tilde x_\ell(0))$ are sampled from the prior of $\left(\mathcal K, X^{N}(0)\right)$, $x_\ell(\cdot)$ is the simulated trajectory of the considered SRN, and $y_\ell(\cdot)$ is the simulated observation. 
A good $C_{\eta}$ should make the associated particle filter close to the exact filter when applied to the simulated data. 
However, this criterion cannot be directly used to optimize $C_{\eta}$ as the exact filter is unavailable.

Alternatively, we train the hyperparameter by the accuracy of the inferred parameters and the diversity of particles.
This strategy first requires the parameter identifiability assumption meaning that given long-time observations, the posterior of the model parameter is densely distributed around the true value.
This assumption is not very restrictive; usually, it can be checked by the frequency spectrum  \cite{song2019frequency} and achieved by optimal experimental design \cite{faller2003simulation}.
According to the parameter identifiability, the accuracy of the RPF can be evaluated by the distance between the inferred model parameter (by the RPF) and the exact parameter.
As the true value of the model parameter is known in the simulated data, we can easily apply this rule to training the filter.
In addition, particle diversity in model parameters should also be considered in hyperparameter tuning, and we evaluate it by the entropy of the particles ($\{ \kappa_j\}_{j=1,\dots, M}$). 
To describe particle diversity more accurately, we truncate the value of particles to four decimal places, so that two particles in close distance are regarded as identical ones.
In conclusion, the optimal hyperparameter, denoted by $C^*_{\eta}$, which balances accuracy and diversity, can be obtained by solving the following optimization problem:
\small{
\begin{align}\label{eq. performance of C eta}
&C^*_{\eta} = \\
& \notag \mathop{\arg\min}\limits_{C_{\eta}>0} 
\sum_\ell  \mathbb E \left[\left. {\bar \pi_{M,t_L,C_{\eta}}\left( \frac{\left\|\kappa-\tilde \kappa_\ell\right\|^2_2}{ \left\| \tilde \kappa_\ell \right\|^2_2} \right)}+\frac{1}{\exp\left( \mathcal{E}_{\pi_{M,t_L,C_{\eta}}}\right) }\right|~ Y^{N,\gamma_1}(t_i)  = \tilde y_{\ell}(t_i), ~ \forall t_i \leq t_{L}   \right], 
\end{align}}where $t_L$ is a large time point, $\bar \pi_{M,t_L,C_{\eta}}$ represents the RPF established in \Cref{def RPFs for SCRNs} with the hyperparameter $C_{\eta}$, and $\mathcal{E}_{\pi_{M,t_L,C_{\eta}}}$ is the entropy of the particles $\{\bar \kappa_j\}_{j=1,\dots,M}$ at the time $t_L$.
In this formula, the term $\bar \pi_{M,t_L,C_{\eta}}\left( \left\|\kappa-\tilde \kappa_\ell\right\|^2_2/ \left\| \tilde \kappa_\ell \right\|^2_2 \right)$ evaluates the relative error between the inferred model parameter and the exact one in the simulated data, and $\exp( \mathcal{E}_{\pi_{M,t_L,C_{\eta}}})$ evaluates the particle diversity.
Note that this optimization problem is nonconvex; to ease the computational burden, we can find an approximation to $C^*_{\eta}$ by optimizing the above objective function on a finite-size grid. 

\begin{remark}
	Even though we only search for the optimal hyperparameter in a finite set, the above optimization problem still requires a lot of computational resources, because for each $(\tilde \kappa_{\ell}, \tilde x_\ell(\cdot),\tilde y_\ell(\cdot))$ it needs to generate many simulated trajectories to calculate the filter.
    We argue that these computational resources spent on hyperparameter training are worthwhile. 
    First, the optimal hyperparameter can make the filter more accurate than that using a random $C_{\eta}$.  
    Second, the training is performed off-line, so it will not increase the computational burden of the RPF when solving a filtering problem online.
\end{remark}

\section{Numerical examples}{\label{Sec. numerical examples}}

In this section, we illustrate our approach using several numerical examples on different scales and show that the proposed regularized particle filter can significantly outperform the SIR counterpart.
The code applied to perform the analysis is available on GitHub: ``https://github.com/ZhouFang92/Regularized-particle-filter-for-stochastic-reaction-networks".

\subsection{Antithetic integral feedback controller}

The antithetic integral feedback (AIF) controller is a biological module that can achieve robust perfect adaptation for output species in arbitrary noisy bimolecular networks (see \cite{aoki2019universal,briat2016antithetic}). 
In this example, we consider such a gene circuit (see \Cref{fig aic}) that has six reactions (see \Cref{fig aic_chemical_equations}) and four species: the control input species (denoted by $S_1$), the sensing species ($S_2$), the regulated species $(S_3)$, and the report species $(S_4)$.
Here, the dynamics are assumed to follow mass-action kinetics. 
The first three reactions in \Cref{fig aic_chemical_equations} result in an ``integral controller" where the integrator is $\mathbb E\left[X_1(t)-X_2(t)\right]$ and ensures that the mean copy number of $S_3$ at the stationary distribution (if it exists) is ${k_1}/{k_2}$ regardless of any other parameters \cite{briat2016antithetic}.
In an \textit{E. coli} cell, such a network can be realized by choosing the sigma factor and anti-sigma factor as $S_1$ and $S_2$, respectively \cite{aoki2019universal}.
Here, we consider a filtering problem for this system in which $S_4$ is a fluorescent reporter and measured under a microscope, and our goal is to infer the dynamics of the regulated species ($S_3$) and its set point $k_1/k_2$.

\begin{figure}[htbp]
	\centering
	\subfigure[Gene circuit: $S_4$ is recorded.]{                 
		\centering                                                         
		\includegraphics[width= 0.45\textwidth]{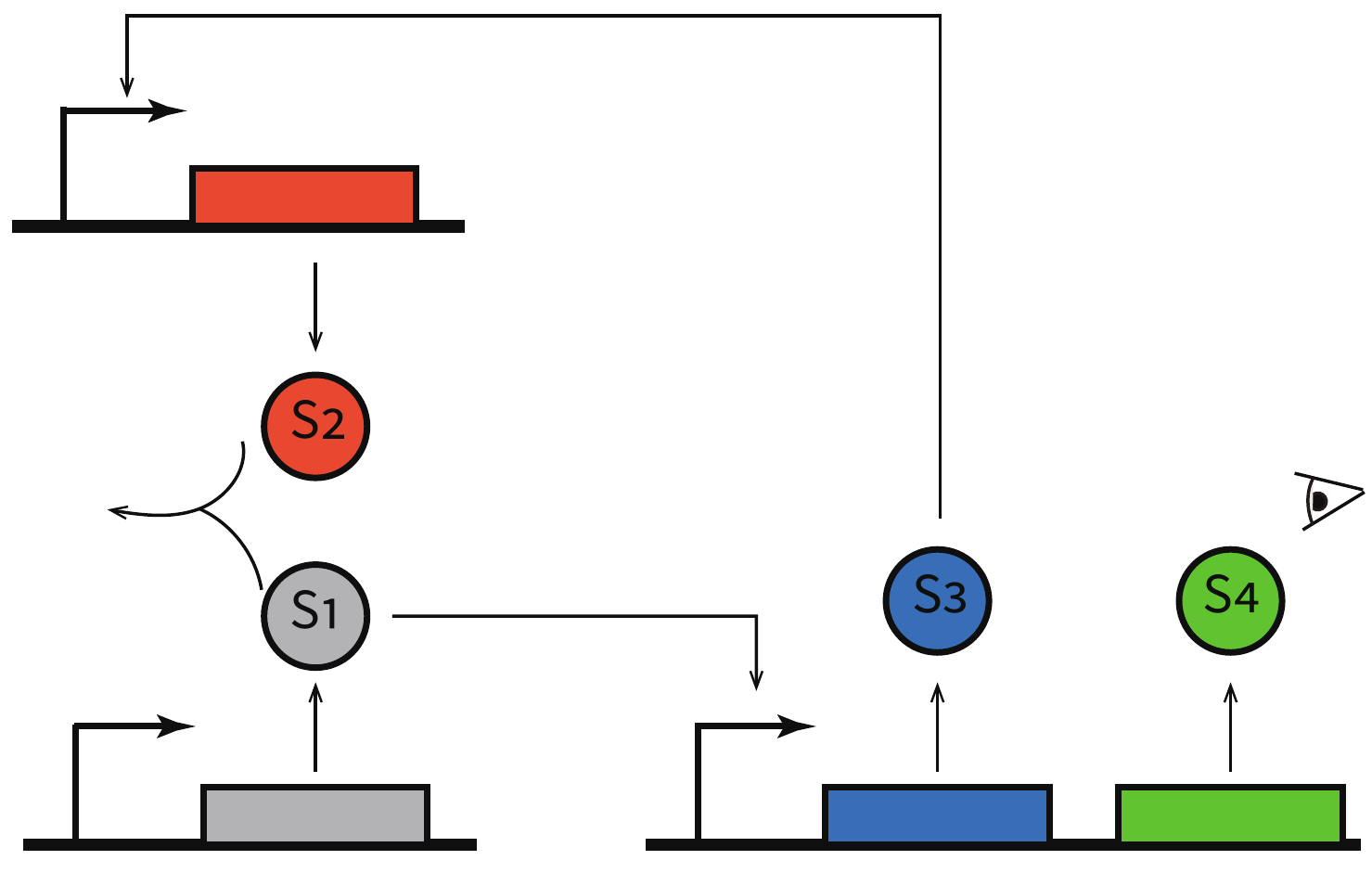}   
		\label{fig aic}            
	}
	\subfigure[Chemical reactions]{
		\vbox{\hsize=0.45\hsize 
			\begin{align*}
			&\emptyset \stackrel{k_1}{\longrightarrow} S_1 && \text{reference} \\
			&S_3  \stackrel{k_2}{\longrightarrow} S_2+ S_3 && \text{measurement} \\
			&S_1+S_2 \stackrel{k_3}{\longrightarrow} \emptyset && \text{comparison} \\
			&S_1  \stackrel{k_4}{\longrightarrow} S_1+S_3 +S_4 && \text{actuation} \\
			&S_3  \stackrel{k_5}{\longrightarrow} \emptyset      && \text{degradation}\\
			&S_4  \stackrel{k_6}{\longrightarrow} \emptyset      && \text{degradation}
			\end{align*}}
		\label{fig aic_chemical_equations}
	}
	\caption{{\bf Antithetic integral feedback controller.} (a) The diagram of an antithetic integral feedback controller  where $S_1$ and $S_2$ form an integral controller with the integrator $\mathbb E\left[X_1(t)-X_2(t)\right]$ ensuring the mean copy number of $S_3$ to be constant at the stationary distribution.   
		$S_4$ is the report species that is visible and measured by a microscope.
		(b) The chemical reactions contained in the antithetic integral feedback controller. Its dynamics is assumed to follow mass-action kinetics with reaction constants $k_i$ ($i=1,\dots,6$). 
		With these parameters, the set point of the regulated species $S_3$ is $k_1/k_2$. 
	}
	\label{fig. AIC_four_species}
\end{figure}

\begin{table}[htbp]
	\centering
	\begin{tabular}{|l|l|l|l|}
		\hline
		\multicolumn{1}{|c|}{Model parameter} &  \multicolumn{1}{|c|}{Exponent} &  \multicolumn{1}{|c|}{Rescaled parameter}  & \multicolumn{1}{|c|}{Initial condition}\\
		\multicolumn{1}{|c|}{$\left(\text{min}^{-1}\right)$} & & \multicolumn{1}{|c|}{$\left(\text{min}^{-1}\right)$}  & 
		\\
		\hline
		$k_1\sim \mathcal U (0.1, 1)$ & $\beta_1=0$ & $k'_1\sim\mathcal U (0.1, 1)$ &
		$X^N_1(0) \sim \mathcal{UI}(0,10)$~~~~ \\
		$k_2\sim \mathcal U (0.1, 1)$ & $\beta_2=0$ & $k'_2\sim\mathcal U (0.1, 1)$ &
		$X^N_2(0)\sim \mathcal{UI}(0,10)$~~~~ \\
		$k_3=0.2$~~ & $\beta_3=0$ & $k'_3=0.2$  &
		$X^N_3(0)\sim \mathcal{UI}(0,20)$~~~~ \\
		$k_4=0.5$ &  $\beta_4=0$ & $k'_4=0.5$  &
		$X^N_4(0)\sim \mathcal{UI}(0,20)$~~~~ \\
		$k_5=0.7$ &  $\beta_5=0$ & $k'_5=0.7$  & \\
		$k_6=0.3$ &  $\beta_6=0$ & $k'_6=0.3$  & \\
		\hline
	\end{tabular}
	\caption{\noindent {\bf Model parameters of the AIF controller.} $\mathcal{U}$ is the uniform distribution, and $\mathcal{UI}$ is the integer uniform distribution.}
	\label{table ex_AIC model_four_species parameters}
\end{table}

{
\begin{figure}[htbp]
	\centering
	\subfigure[Particle filters for the AIF controller.]{
		\centering
		\includegraphics[width= 0.9 \textwidth]{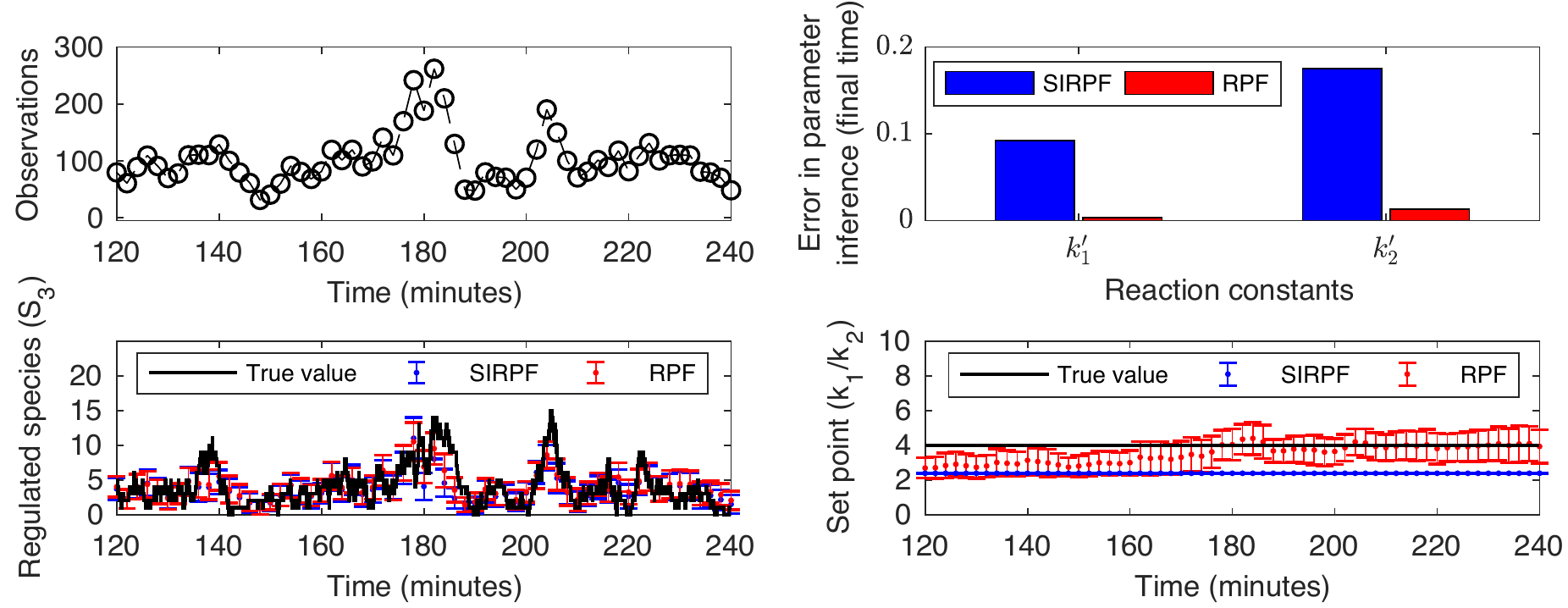}
		\label{fig aic_trajectory}
	}
	\subfigure[The distribution of particles at the final time.]{                 
		\centering                                                         
		\includegraphics[height= 0.2\textheight]{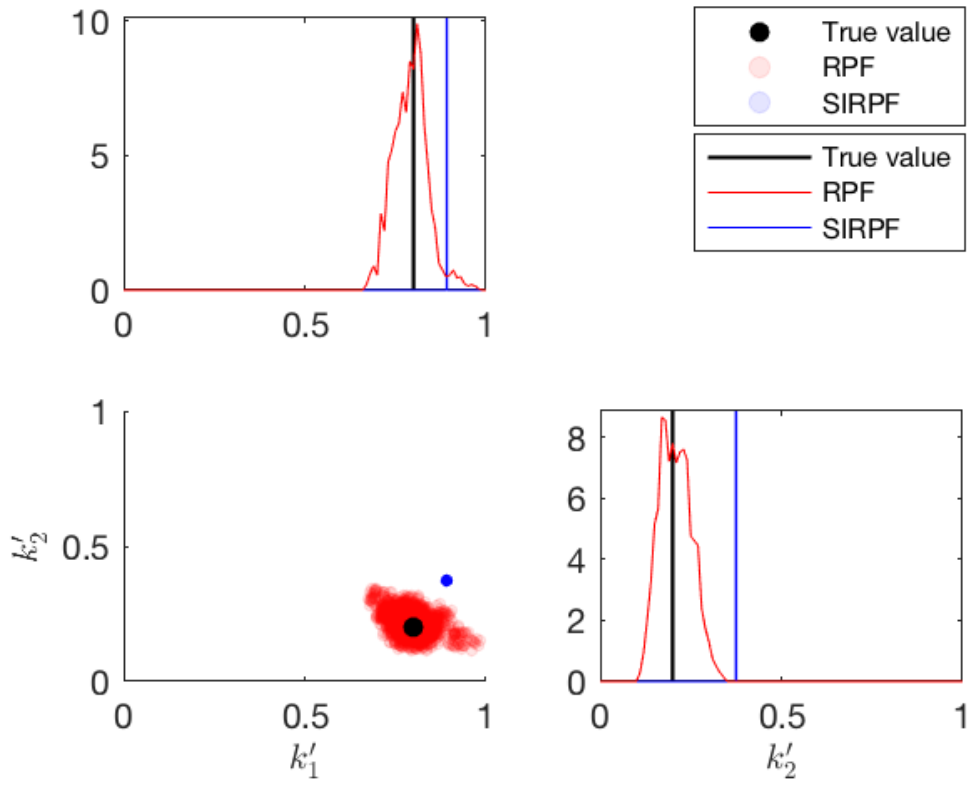}   
		\label{fig aic_final_particles}            
	}
	\hfill
	\subfigure[The evolution of particles over time.]{                 
		\centering                                                         
		\includegraphics[height= 0.2 \textheight]{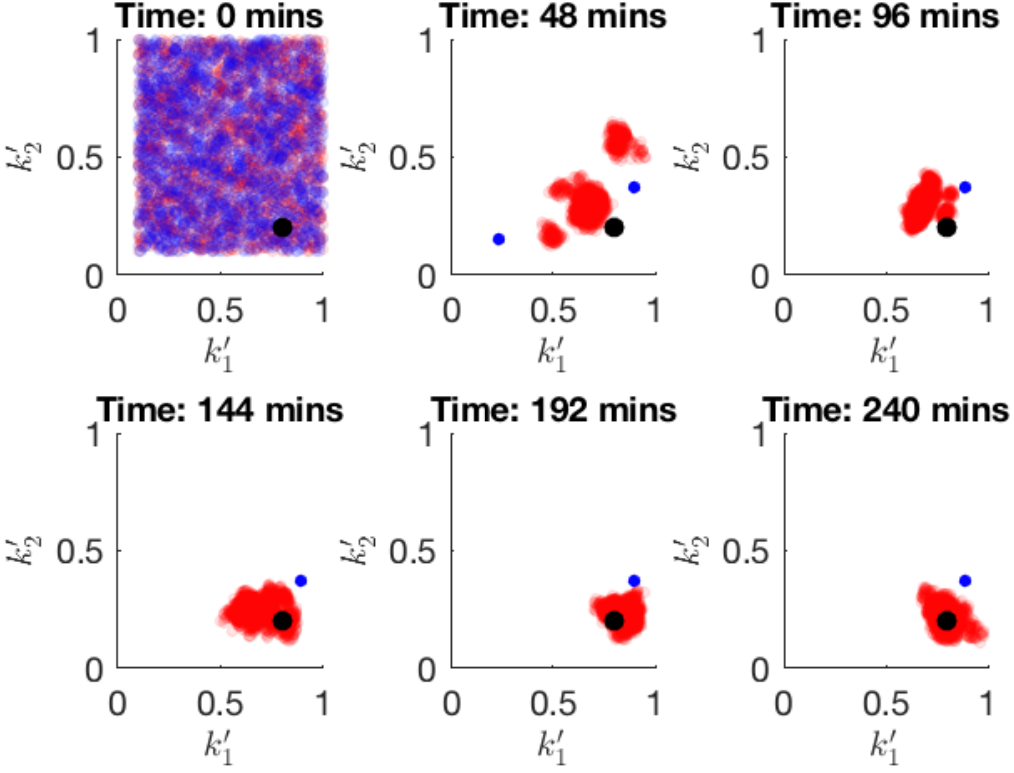}      
		\label{fig aic_particle_evolution}                
	}
	\caption{{\bf Numerical results of the AIF controller.} ``RPF" stands for the regularized particle filter, and ``SIRPF" stands for the sequential importance resampling particle filter. 
		Panel (a) shows the performance of both filters in estimating the regulated species $S_3$ and its set point $k_1/k_2$. 
		It tells that both filters can infer the dynamics of $S_3$ accurately, but the RPF has a much better performance in estimating the set point and model parameters. 
		Panel (b) presents the final time distribution of the particles associated with these filters in the parameter coordinate, and panel (c) shows the evolution of these particles over time.   
		These two panels tell that sample degeneracy happens to the SIRPF very early in time (around 48 minutes), and, therefore, the SIRPF cannot adjust its estimate of model parameters according to the new observations afterward.
		This is the reason why the SIRPF fails to infer model parameters accurately.
		In contrast, the particle diversity in the RPF is always large, and this enables the RPF to constantly refine its estimate of model parameters according to new observations.
		Consequently, the RPF can infer model parameters and the set point accurately. 
	}
	\label{fig aic_result}
\end{figure}

\begin{table}[htbp]
	\centering
	\begin{tabular}{|c|c|c|c|c|c|c|c|c|c|c|c|}
		\hline 
		$C_{\eta}$ & 0 & 0.2 & 0.4 & 0.6 & 0.8 & 1.0 & 1.2 & 1.4 & 1.6 & 1.8* & 2\\
		\hline
		Score   & 1.00 & 0.14& 0.11& 0.11& 0.11 & 0.10& 0.09& 0.09& 0.10& 0.08*& 0.09\\
		\hline
	\end{tabular}
	\caption{\noindent {\bf The performance of different $C_{\eta}$ (the AIF controller):} $t_L=240 ~(\rm{mins})$.}
	\label{table exAIC_four_species hyperparameters}
\end{table}
}

In this example, we set the scaling factor ($N$) to be 100 and the abundance coefficients ($\alpha_1$, $\alpha_2$, $\alpha_3$, and $\alpha_4$) to be zero, meaning that the cell system has very few copies of the involved species.
Also, we assume the model parameters to satisfy \Cref{table ex_AIC model_four_species parameters}, where all the rate constants are known except for $k_1$ and $k_2$.
The observations are collected every 2 minutes and satisfy
\begin{equation*}
Y^{\gamma_1}(t_i) = 10 X_4(t_i) \wedge 1000 + W(t_i) \qquad i\in \mathbb N_{>0},
\end{equation*}
in which $t_i=2i$, and $\{W(t_i)\}_{i\in\mathbb N_{>0}}$ is a sequence of mutually independent standard Gaussian random variables.
In this setting, we can calculate that $\gamma_1=0$, and the most proper RPF to solve this filtering problem is the one utilizing the exact simulation algorithm in the sampling step (i.e., the first filter in \Cref{def RPFs for SCRNs}). 
We set the particle population of the filter to be 2000 and trained the RPF following the rule \eqref{eq. performance of C eta}; see \Cref{table exAIC_four_species hyperparameters} for the training result.
Finally, we applied the RPF with the optimal hyperparameter to a simulated process and compared its result with the SIRPF; the numerical results are shown in \Cref{fig aic_result}.

\Cref{fig aic_trajectory} shows that both filters can infer the dynamics of the regulated species accurately; however, the RPF has a superior performance in estimating the model parameters and the set point 
thanks to the larger particle diversity of the RPF (\Cref{fig aic_final_particles} and \Cref{fig aic_particle_evolution} ).
Moreover, the sample degeneracy of the SIRPF occurs very early in time (\Cref{fig aic_particle_evolution}), suggesting that the SIRPF fails to represent the posterior of the model parameters very quickly after the starting time.  
All these facts suggest that the RPF outperforms the SIRPF in this numerical example.

\subsection{Gene transcription model}\label{subsection gene transcription model}
We then consider a gene transcription model (\Cref{fig. GT}) consisting of four reactions (see \Cref{fig GT_chemical_equations}) and three species: the inactivated gene (denoted by $S_1$), the activated gene ($S_2$), and the transcribed RNA ($S_3$).
Moreover, we assume that the RNA contains stem-loops recognized and bound by fluorescent reporters so that the produced RNAs can be visualized in a living cell. 
This model resembles the gene circuit constructed in \cite{rullan2018optogenetic} for investigating and controlling  transcriptional dynamics.
In this example, our goal is to infer both hidden dynamic states and model parameters from the history of the RNA signal. 

\begin{figure}[htbp]
	\centering
	\subfigure[The circuit: the RNA is recorded.]{                 
		\centering                                                         
		\includegraphics[width= 0.38\textwidth]{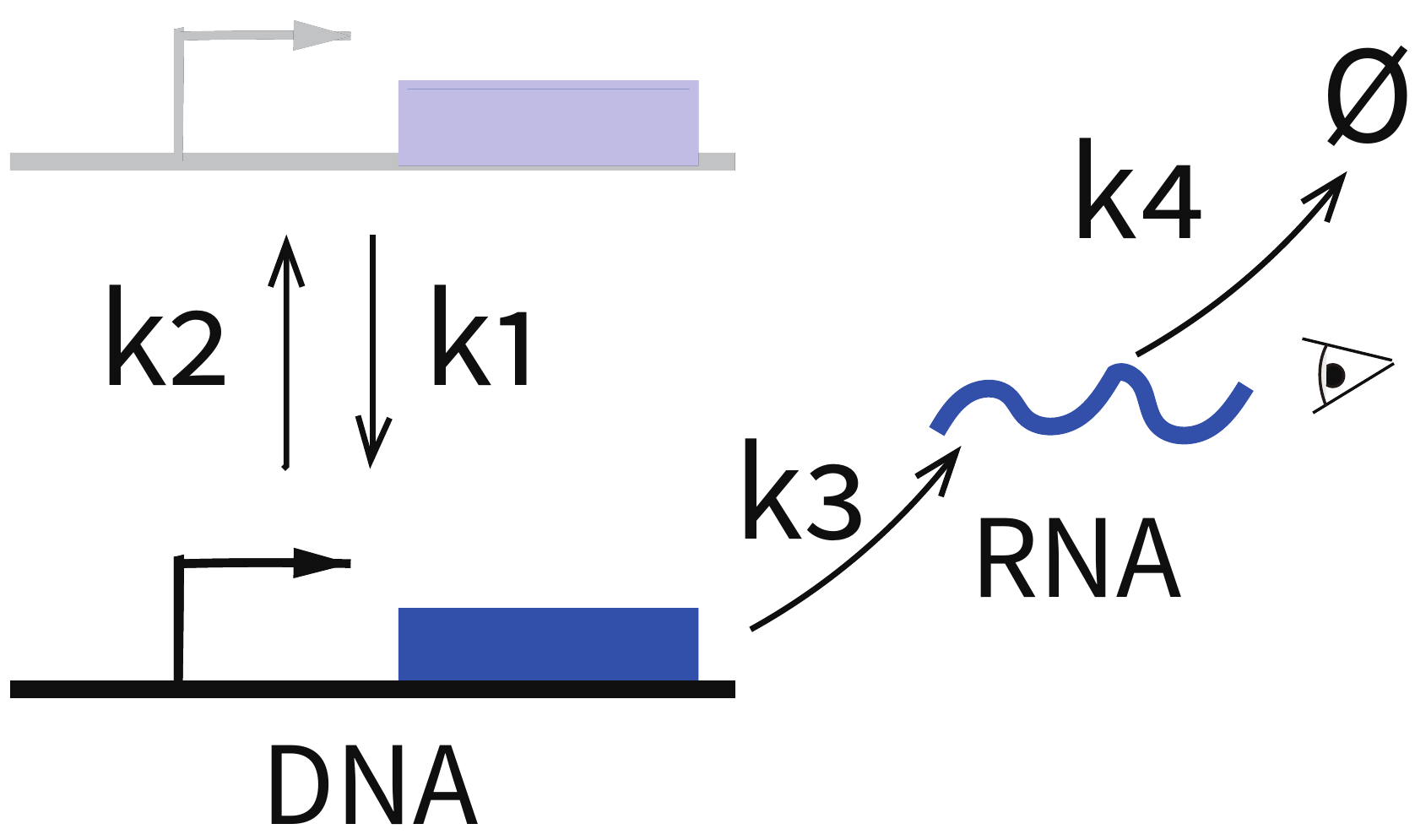}   
		\label{fig GT circuit}            
	}
	\hfill
	\subfigure[Chemical reactions]{
		\vbox{\hsize=0.5\hsize 
			\begin{align*}
			&S_1 \stackrel{k_1}{\longrightarrow} S_2 && \text{gene activation} \\
			&S_2  \stackrel{k_2}{\longrightarrow} S_1 && \text{gene inactivation} \\
			&S_2 \stackrel{k_3}{\longrightarrow} S_2+S_3 && \text{transcription} \\
			&S_3  \stackrel{k_4}{\longrightarrow} \emptyset      && \text{degradation}
			\end{align*}}
		\label{fig GT_chemical_equations}
	}
	\caption{{\bf Gene transcription model.} (a) The diagram of a gene transcription circuit where the gene has ``on" and ``off" states and transcribes message RNAs only in the ``on" state.
	The message RNAs are visible, and their light intensity signal is measured by a microscope. 
	(b) The chemical reactions contained in the gene transcription model, where $S_1$, $S_2$, and $S_3$ are the inactivated gene, the activated gene, and the message RNA, respectively. All the reactions follow mass-action kinetics with reaction constants $k_i$ ($i=1,\dots,4$).
    }
	\label{fig. GT}
\end{figure}

\begin{table}[htbp]
	\centering
	\begin{tabular}{|l|c|c|l|}
		\hline
		\multicolumn{1}{|c|}{Model parameter} & Exponent & Rescaled parameter  & \multicolumn{1}{|c|}{Initial condition}\\
		$\left(\text{min}^{-1}\right)$ & & $\left(\text{min}^{-1}\right)$  & 
		\\
		\hline
		$k_1\sim \mathcal U (0.03, 0.3)$ & $\beta_1=0$ & $k'_1\sim\mathcal U (0.03, 0.3)$ &
		$X^N_1(0) \sim$ Bern($1/2$)~~~~ \\
		$k_2\sim \mathcal U (0.03, 0.3)$ & $\beta_2=0$ & $k'_2\sim\mathcal U (0.03, 0.3)$ &
		$X^N_2(0)=1-X^N_1(0)$~~~~ \\
		$k_3\sim \mathcal U (10, 100)$~~ & $\beta_3=1$ & $k'_3\sim\mathcal U (0.2, 2)$~~~~  &
		$X^N_3(0) \sim \frac{X^N_2(0)\text{Poiss(50)}}{50}$ \\
		$k_4\sim \mathcal U (0.2, 2)$~~~~ & $\beta_4=0$ & $k'_4\sim\mathcal U (0.2, 2)$~~~~ & \\
		\hline
	\end{tabular}
	\caption{\noindent {\bf Model parameters of the gene transcription model.} $\mathcal U$ is the uniform distribution, ``{\rm{Bern($\cdot$)}}" is the Bernoulli distribution, and ``{\rm{Poiss}}" is the Poisson distribution.}
	\label{table exGT_tau-leaping model parameters}
\end{table}

\begin{figure}[htbp]
	\centering
	\subfigure[Particle filters for the gene transcription model.]{
		\centering
		\includegraphics[width= 0.9 \textwidth]{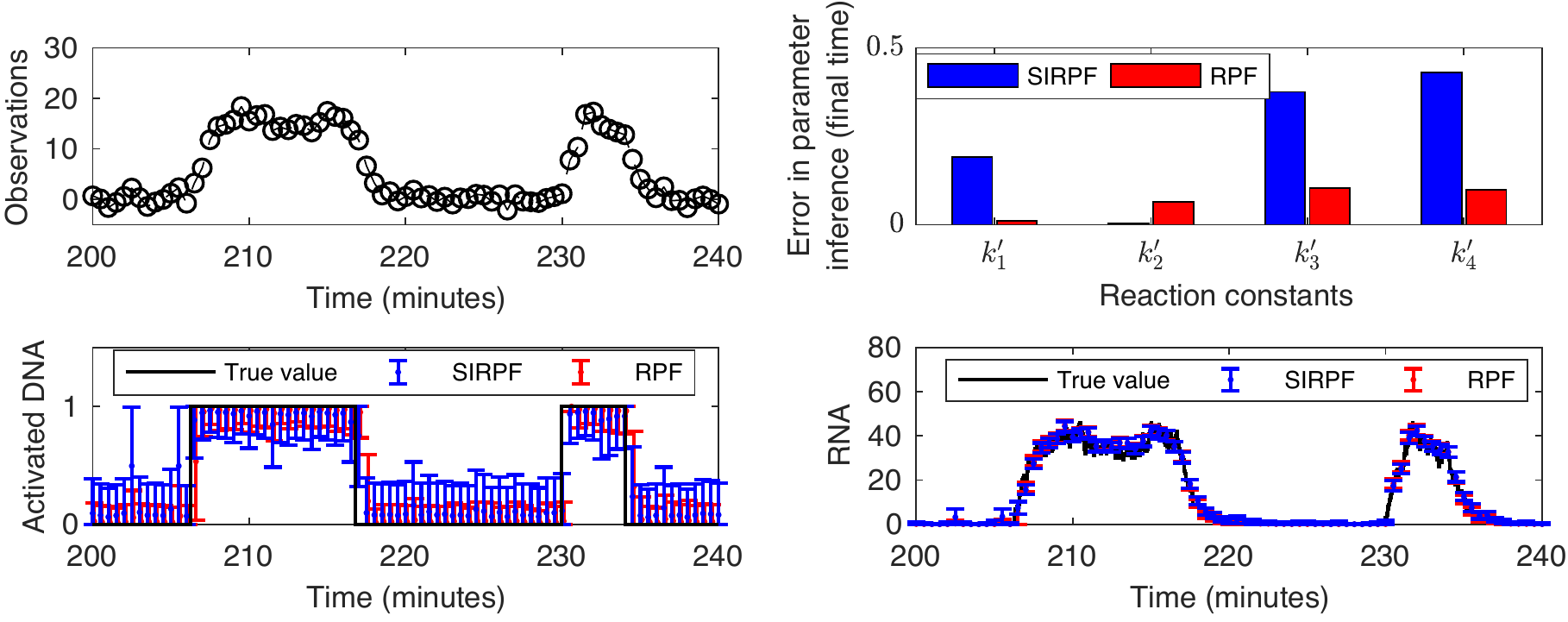}
		\label{fig GT_trajectory}
	}
	\subfigure[The distribution of particles at the final time.]{                 
		\centering                                                         
		\includegraphics[height= 0.2\textheight]{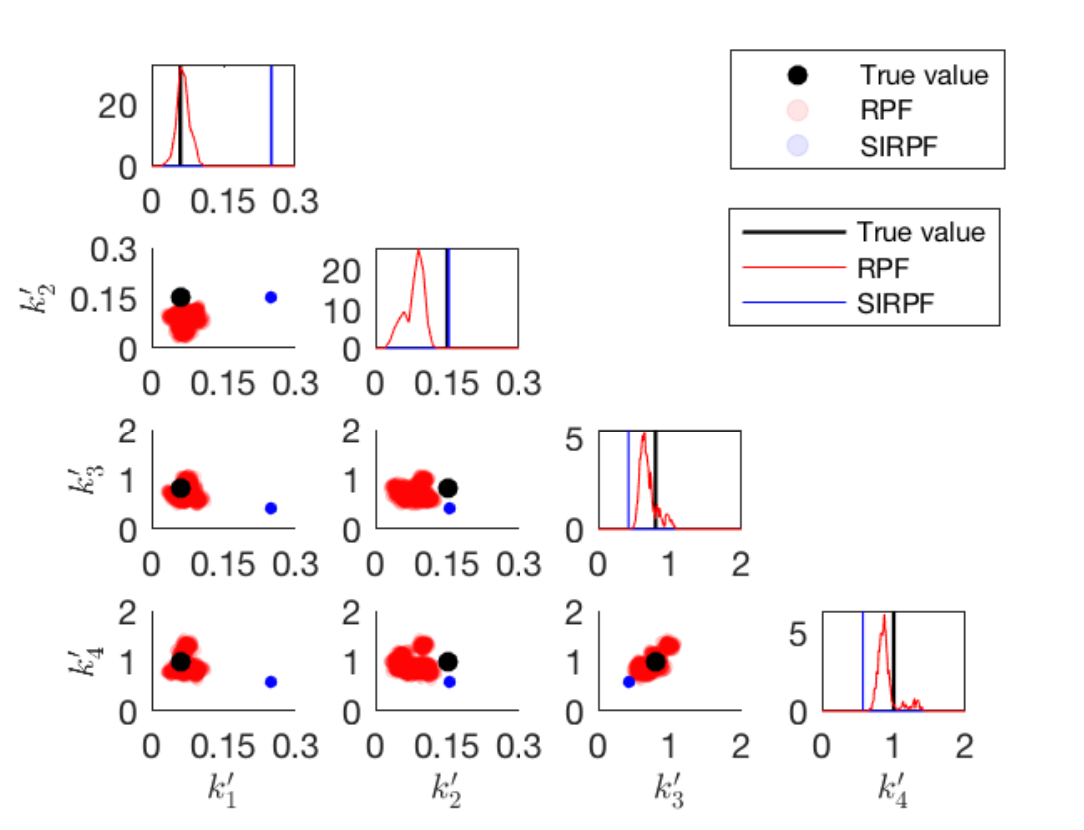}   
		\label{fig GT_final_particles}            
	}
	\hfill
	\subfigure[The evolution of particles over time.]{                 
		\centering                                                         
		\includegraphics[height= 0.2 \textheight]{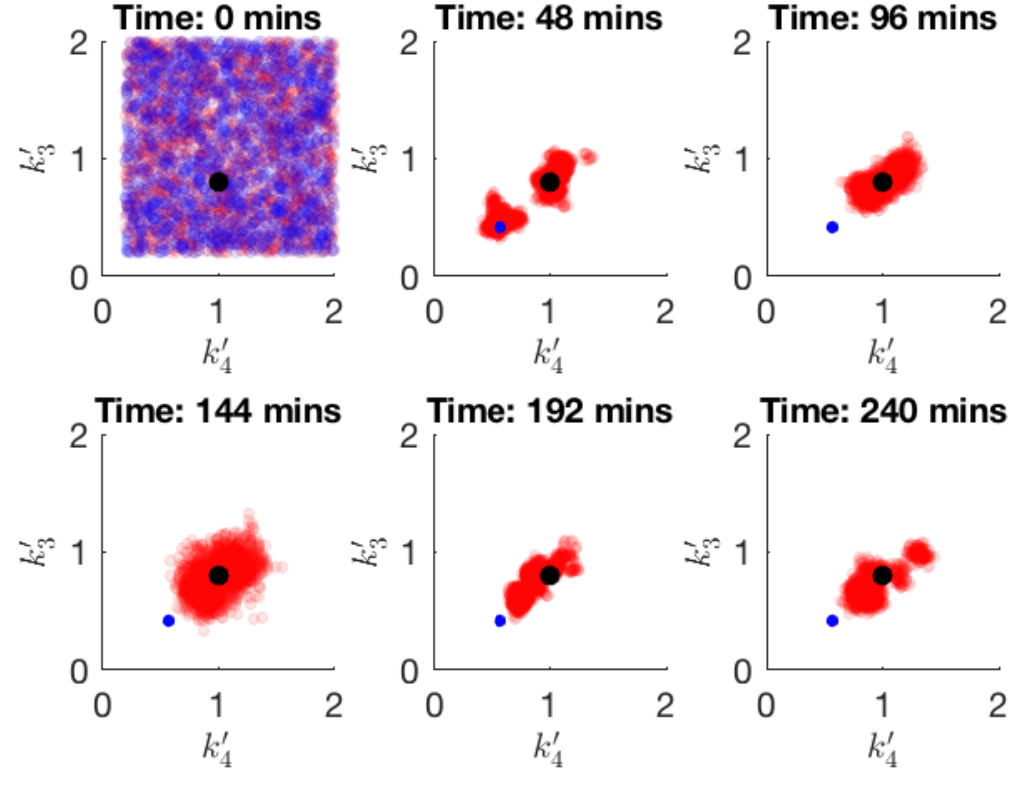}      
		\label{fig GT_particle_evolution}                
	}
	\caption{{\bf Numerical results of the gene transcription model.} ``RPF" stands for the regularized particle filter, and ``SIRPF" stands for the sequential importance resampling particle filter.
		Panel (a) shows the performance of two filters in estimating the system states and model parameters.
		Specifically, both filters infer the dynamics of DNAs and RNAs accurately, but the RPF has a better performance in estimating model parameters. 
		Panel (b) presents the final time distribution of the particles associated with these filters in the parameter coordinate, and panel (c) shows the time evolution of these particles. 
		Similar to the previous example, these two panels tell that sample degeneracy happens to the SIRPF very early in time (before 48 minutes), which disables the SIRPF to adjust its estimate of model parameters according to the new observations afterward. 
		In contrast, the RPF always has a large particle diversity and is able to refine its estimate of model parameters over time. 
		This explains why the RPF outperforms the SIRPF in parameter inference. 
		However, note that the RPF does not infer every parameter accurately; its estimate of $k'_2$ somehow deviates from the exact value (though the deviation is not too large). 
	}
	\label{fig GT_result}
\end{figure}

\begin{table}[htbp]
	\centering
	\begin{tabular}{|c|c|c|c|c|c|c|c|c|c|c|c|}
		\hline 
		$C_{\eta}$ & 0 & 0.1 & 0.2 & 0.3 & 0.4 & 0.5 & 0.6 & 0.7 & 0.8 & 0.9* & 1\\
		\hline
		Score   & 1.08 & 0.17& 0.13& 0.10& 0.10& 0.09 & 0.09& 0.09& 0.09& 0.08*& 0.09\\
		\hline
	\end{tabular}
	\caption{\noindent {\bf The performance of different $C_{\eta}$ (the gene transcription model):} $t_L=240 ~(\rm{mins})$.}
	\label{table exGT_tau-leaping hyperparameters}
\end{table}

For the considered gene transcription model, we set $N=50$, $\alpha_1=\alpha_2=0$, and $\alpha_3=1$, meaning that the system contains very few copies of DNAs but tens of mRNA copies.
Also, we assume the dynamics of the system to follow mass-action kinetics and the model parameters to satisfy \Cref{table exGT_tau-leaping model parameters} where all reaction constants are unknown. 
The observations are measured every 0.5 minute and satisfy
\begin{equation*}
Y^{\gamma_1}(t_i) = 20 X^{N,\gamma_1}_3(t_i) \wedge 1000 + W(t_i) \qquad i\in \mathbb N_{>0},
\end{equation*}
where $\gamma_1=0$, $t_i=0.5i$, 1000 is the measurement range, and $\{W(t_i)\}_{i\in\mathbb N_{>0}}$ is a sequence of mutually independent standard Gaussian random variables. 
Here, we utilized the tau-leaping based RPF to solve the associated filtering problem.
Again, we set the particle population to be 2000 and trained the filter following the rule \eqref{eq. performance of C eta}; see \Cref{table exGT_tau-leaping hyperparameters} for the training results.
Finally, we applied the filter to a simulated process, whose results are present in \Cref{fig GT_result}.

Similar to the conclusion obtained in the previous subsection, \Cref{fig GT_result} shows that the RPF has a superior performance in estimating model parameters due to its larger particle diversity, and the sample degeneracy of the SIRPF occurs very early in time. 
These facts imply that RPF outperforms the SIRPF in this example.
However, we note that the RPF does not infer every parameter accurately (see $k'_2$ in \Cref{fig GT_final_particles}); this might be due to the fact that particle population is not large enough in this example.
Also, both filters have similar performances in estimating the gene dynamics (\Cref{fig GT_trajectory}) due to the simple structure of the system: once the observation increases/decreases sharply, the filter can learn that the gene is turned on/off whatever its estimates of model parameters.

\subsection{Transcription regulation network}\label{subsection transcription regulation network}
Finally, we consider a transcription regulation network (\Cref{fig. TR}) where the gene product inhibits the gene expression process by binding to the DNA molecule.
This system consists of six reactions (see \Cref{fig TR_chemical_equations}) and four chemical species: the inactivated gene (denoted by $S_1$), the activated gene ($S_2$), the mRNA ($S_3$), and the protein ($S_4$).
Also, we assume that the protein is fluorescent and, therefore, can be measured by a microscope.
In this example, our goal is to infer both the hidden dynamic states and model parameters from the time-course data.
\begin{figure}[htbp]
	\centering
	\subfigure[The circuit: the protein is recorded.]{                 
		\centering                                                         
		\includegraphics[width= 0.45\textwidth]{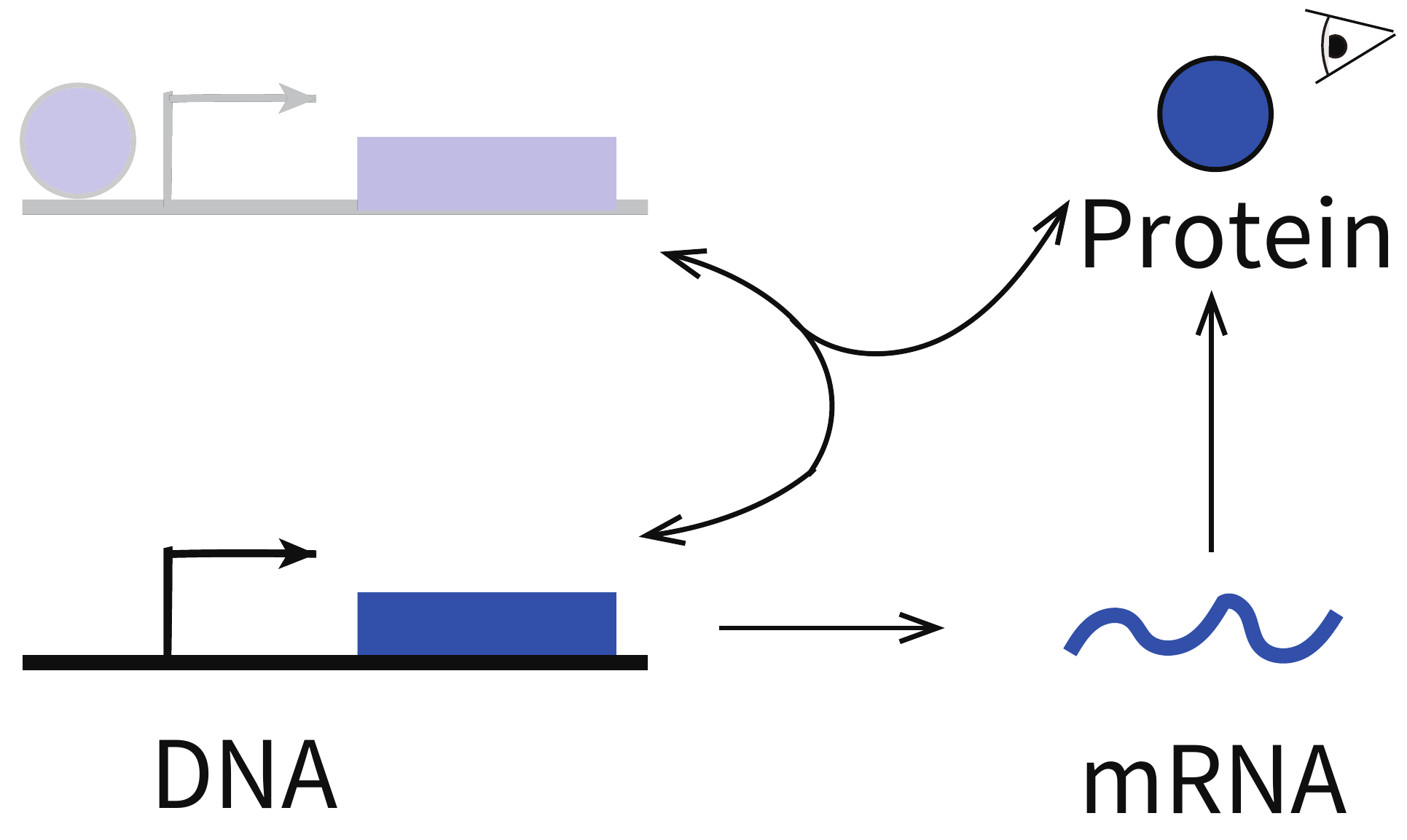}   
		\label{fig TR circuit}            
	}
	\hfill
	\subfigure[Chemical reactions]{
		\vbox{\hsize=0.45\hsize 
			\begin{align*}
			&S_2+S_4 \stackrel{k_1}{\longrightarrow} S_1 && \text{gene inactivation} \\
			&S_1  \stackrel{k_2}{\longrightarrow} S_2+S_4 && \text{gene activation} \\
			&S_2 \stackrel{k_3}{\longrightarrow} S_2+S_3 && \text{transcription} \\
			&S_3  \stackrel{k_4}{\longrightarrow} \emptyset      && \text{degradation} \\
			&S_3 \stackrel{k_5}{\longrightarrow} S_3+S_4 && \text{translation} \\
			&S_4 \stackrel{k_6}{\longrightarrow} \emptyset  && \text{degradation} 
			\end{align*}}
		\label{fig TR_chemical_equations}
	}
	\caption{{\bf Transcription regulation network.} (a) The diagram of a transcription regulation network, where the gene product (the protein) inhibits transcription by silencing the gene. The protein contains a fluorescent tag so that it is visible and measured by a microscope. 
	(b) The chemical reactions contained in the transcription regulation network, where $S_1$, $S_2$, $S_3$, and $S_4$ are the inactivated gene, the activated gene, the mRNA, and the protein, respectively. 
	All the reactions follow the mass-action kinetics with reaction constants $k_i$ ($i=1,\dots, 6$).
    }
	\label{fig. TR}
\end{figure}

\begin{table}[h]
	\centering
	\begin{tabular}{|l|c|c|l|}
		\hline
		\multicolumn{1}{|c|}{Model parameter} & Exponent & Rescaled parameter  & \multicolumn{1}{|c|}{Initial condition}\\
		$\left(\text{min}^{-1}\right)$ & & $\left(\text{min}^{-1}\right)$  & 
		\\
		\hline
		$k_1\sim \mathcal U (0.03, 0.3)$ & $\beta_1=0$ & $k'_1\sim\mathcal U (0.03, 0.3)$ &
		$X^N_1(0) \sim$ Bern($1/2$)~~~~ \\
		$k_2\sim \mathcal U (0.03, 0.3)$ & $\beta_2=0$ & $k'_2\sim\mathcal U (0.03, 0.3)$ &
		$X^N_2(0)=1-X^N_1(0)$ \\
		$k_3\sim \mathcal U (0.1, 1)$~~ & $\beta_3=0$ & $k'_3\sim\mathcal U (0.1, 1)$~~~~  &
		$X^N_3(0) \sim \text{Poiss(10)}$ \\
		$k_4\sim \mathcal U (0.1, 1)$~~~~ & $\beta_4=0$ & $k'_4\sim\mathcal U (0.1, 1)$~~~~ & 
		$X^N_4(0) \sim \frac{\text{Poiss}(10^4)}{10^4}$\\
		$k_5\sim \mathcal U (10^3, 10^4)$~~~~ & $\beta_5=1$ & $k'_5\sim\mathcal U (0.1, 1)$~~~~ & \\
		$k_6\sim \mathcal U (0.1, 1)$~~~~ & $\beta_6=0$ & $k'_6\sim\mathcal U (0.1, 1)$~~~~ & \\
		\hline
	\end{tabular}
	\caption{\noindent {\bf Model parameters of the transcription regulation model}:  $\mathcal U$ is the uniform distribution, ``Bern($\cdot$)" is the Bernoulli distribution, and ``Poiss" is the Poisson distribution.}
	\label{table TR model parameters}
\end{table}

\begin{table}[htbp]
	\centering
	\begin{tabular}{|c|c|c|c|c|c|c|c|c|c|c|c|}
		\hline 
		$C_{\eta}$ & 0 & 0.2 & 0.4 & 0.6 & 0.8 & 1.0 & 1.2 & 1.4 & 1.6 & 1.8* & 2\\
		\hline
		Score   & 1.29 & 0.41& 0.37& 0.34& 0.37& 0.36 & 0.36& 0.36& 0.35& 0.33*& 0.36\\
		\hline
	\end{tabular}
	\caption{\noindent {\bf The performance of different $C_{\eta}$ (the transcription regulation network):} $t_L=240 ~(\rm{mins})$.}
	\label{table exTR_PDMP hyperparameters}
\end{table}

\begin{figure}[htbp]
	\centering
	\subfigure[Particle filters for the gene transcription model.]{
		\centering
		\includegraphics[width= 0.9 \textwidth]{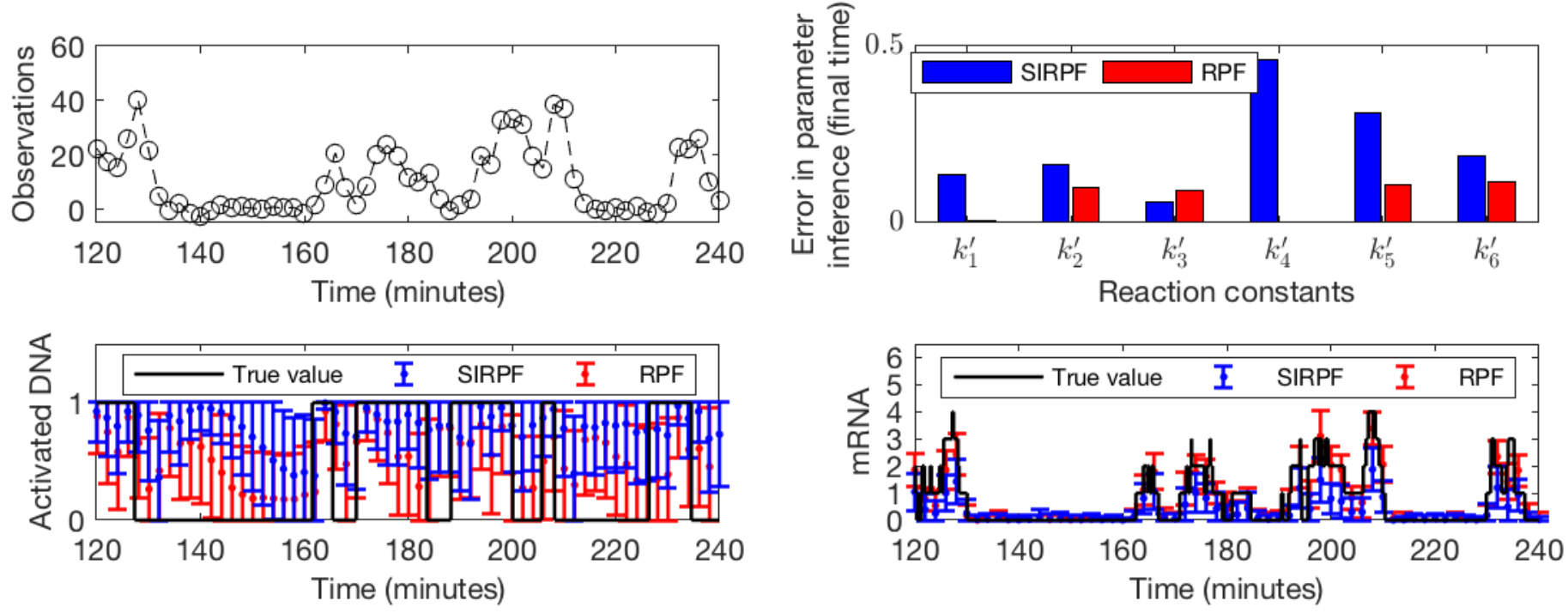}
		\label{fig TR_trajectory}
	}
	\subfigure[The distribution of particles at the final time.]{                 
		\centering                                                         
		\includegraphics[width= 0.45\textwidth]{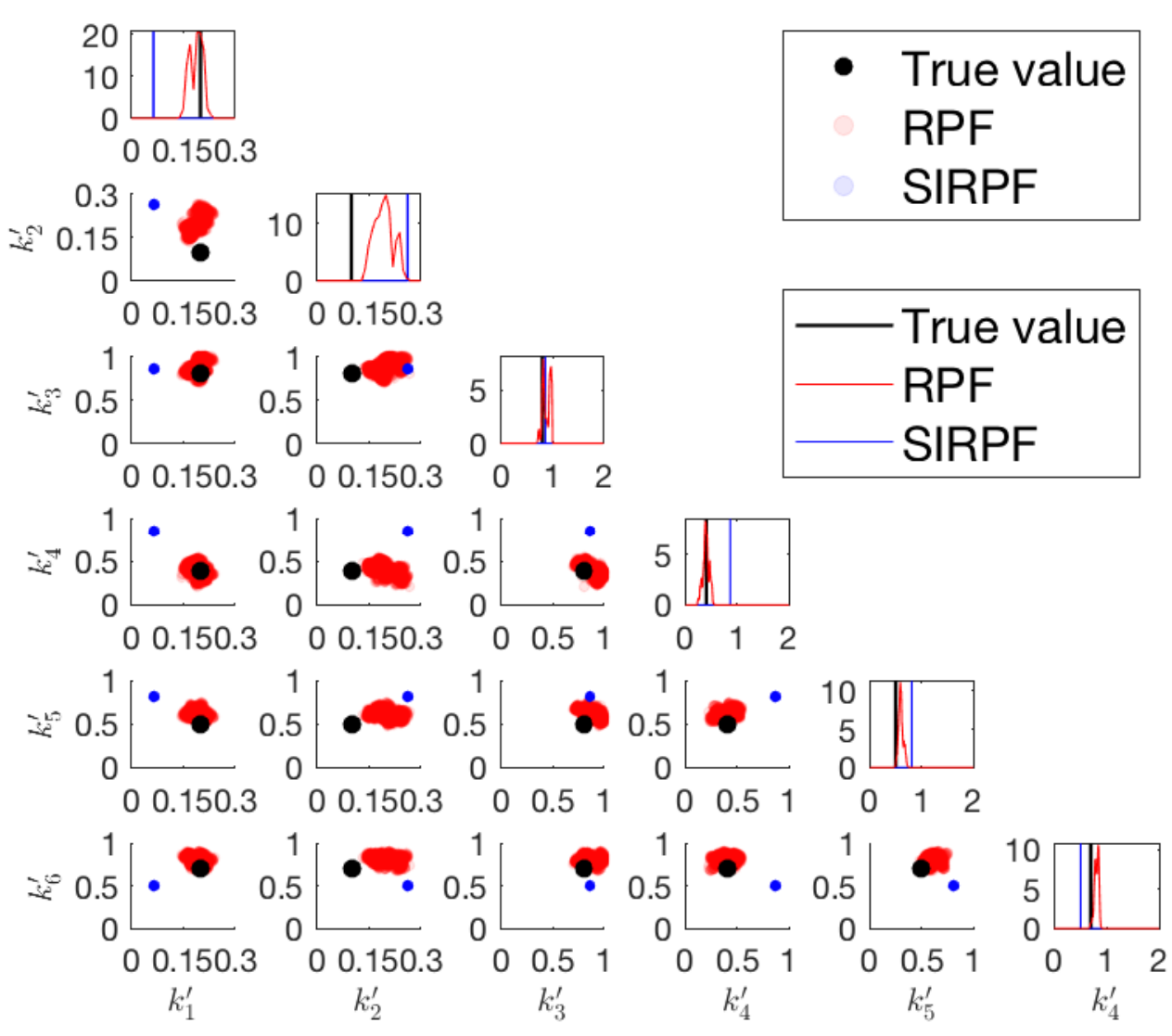}   
		\label{fig TR_final_particles}            
	}
	\hfill
	\subfigure[The evolution of particles over time.]{                 
		\centering                                                         
		\includegraphics[width= 0.4\textwidth]{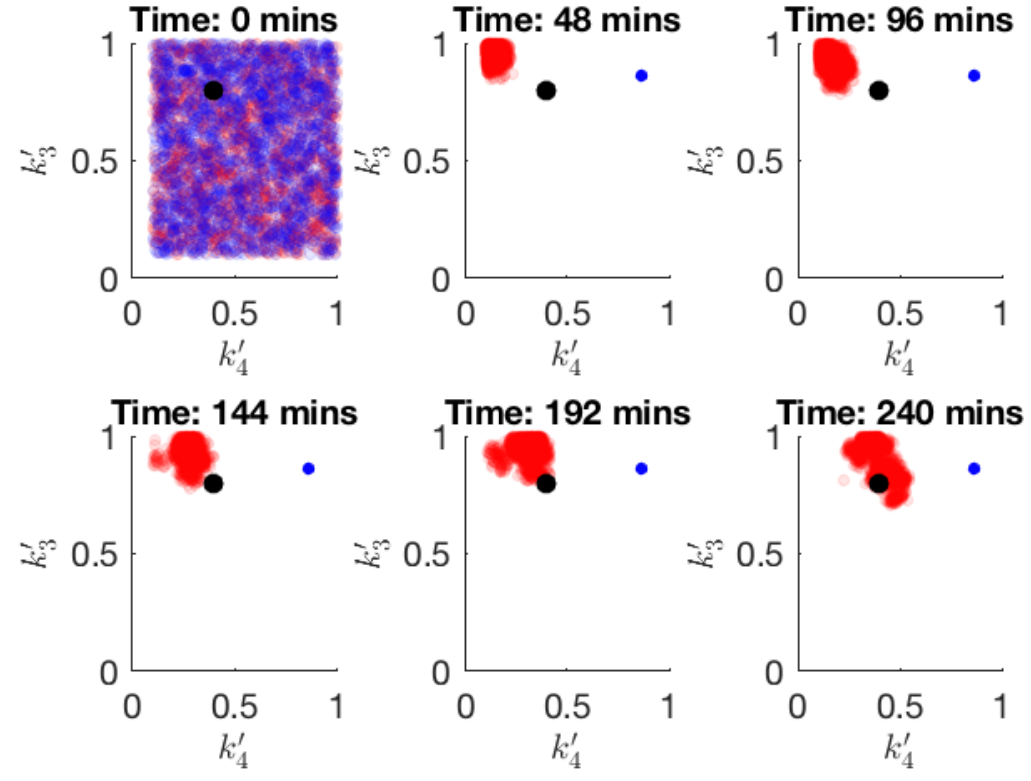}      
		\label{fig TR_particle_evolution}                
	}
	\caption{{\bf Numerical results of the gene transcription model.} ``RPF" stands for the regularized particle filter, and ``SIRPF" stands for the sequential importance resampling particle filter.
		Panel (a) shows the performance of both filters in estimating the hidden system states and model parameters.
		Specifically, RPF has a superior performance in estimating the dynamics of mRNAs and model parameters. (Note that the SIRPF underestimates the mRNA dynamics many times around 200 minutes.)
		However, both filters fail to infer the DNA states accurately.
		Panel (b) shows the final time distribution of the particles associated with these filters in the parameter coordinate, and panel (c) shows the time evolution of these particles. 
		Similar to the previous examples, these two panels tell that sample degeneracy happens to the SIRPF very early in time (before 48 minutes), which disables the SIRPF to adjust its estimate of model parameters according to the new observations afterward. 
		In contrast, the RPF always has a large particle diversity and is enable to refine its estimate of model parameters to approach the exact value. 
		This is the reason why the RPF outperforms the SIRPF in estimating model parameters. 
		Also, note that the RPF does not infer every parameter very accurately; its estimate of $k'_2$ (the reaction constant for gene activation) somehow deviates from the exact value. 
	}
	\label{fig TR_result}
\end{figure}

The modeling of the systems is as follows. We set $N=10^4$, $\alpha_1=\alpha_2=\alpha_3=0$, and $\alpha_4=1$; in other words, the system has very few copies of DNA molecules and mRNA molecules but a large number of proteins. 
Also, we assume that the system dynamics to follow mass-action kinetics and the model parameters to satisfy \Cref{table TR model parameters}. 
Based on it, we can easily calculate that the first timescale $\gamma_1$ equals to 0.
The observations are collected every 2 minutes and satisfy 
\begin{equation*}
Y^{\gamma_1}(t_i) = 20 X^{N,\gamma_1}_4(t_i) \wedge 1000 + W(t_i) \qquad i\in \mathbb N_{>0},
\end{equation*}
where $t_i=2i$, 1000 is the measurement range, and $\{W(t_i)\}_{i\in\mathbb N_{>0}}$ is a sequence of mutually independent standard Gaussian random variables.
In this setting, the most proper filter for the considered system is the RPF using PDMPs in sampling (the third one in \Cref{def RPFs for SCRNs}).
Again, we set the particle population to be 2000 and trained the filter following the rule \eqref{eq. performance of C eta}; see \Cref{table exTR_PDMP hyperparameters} for the training results.
Finally, we applied the filter to a simulated process, whose results are present in \Cref{fig TR_result}.

Similar to the previous examples, the RPF outperforms the SIRPF in this example in estimating model parameters and the mRNA dynamics\footnote{The SIRPF underestimates the mRNA dynamics for many times; see the result around 200 minutes.} due to its large particle diversity  (\Cref{fig TR_result}). 
Moreover, even though the initial estimate is not very accurate, the particles in the RPF can gradually approach the exact value of the model parameters and, consequently, provide an accurate estimate of the model parameters at the final time (\Cref{fig TR_particle_evolution}). 

Despite these successes, we can observe that both filters fail to precisely estimate the gene state (\Cref{fig TR_trajectory}).
This is attributed to the long time delay between the gene dynamics and the protein dynamics, which makes the exact posterior of the gene state very noisy.
Inspired by this observation, we argue that given an arbitrary gene circuit, it is not reasonable to expect the filter to infer every species accurately, especially those to whose changes the fluorescent reporters respond slowly.

\section{Conclusion}\label{conclusion}

The simultaneous inference of hidden dynamic states and model parameters of a stochastic reaction network is an important and difficult task.
One challenge lies in the high computational complexity of the exact simulation method for the stochastic reaction dynamics, which makes the particle filter, a simulation-based algorithm, computationally inefficient. 
Moreover, the sample degeneracy occurring in conventional particle filters has a negative impact on the accuracy of the estimate.
To tackle these problems, this paper introduced regularized particle filters (RPFs) for stochastic reaction networks (SRNs) on different scales using computationally efficient simulation algorithms.
Specifically, the filter chooses the proper simulation algorithm for sampling based on the scale of the system considered so that the computational time is acceptable (\Cref{subsubsection tau-leaping and PDMP} and \cref{def RPFs for SCRNs}).
Also, an artificial evolution step is executed after resampling to mitigate sample degeneracy and cope with the fluctuation of model parameters in practical systems (\cref{remark artificial evolution}).
Furthermore, by parameter sensitivity analyses, we showed that the constructed RPF  is convergent to the exact filter for a broad class of SRNs (\cref{thm main} and \cref{remark covered SCRNs}).
Finally, several numerical examples were shown to illustrate the superior performance of the constructed RPF compared with its sequential-importance-resampling (SIR) counterpart. 

From the numerical examples, we can observe that the constructed RPF also has some drawbacks; for instance, not every parameter can be inferred accurately when the number of unknown parameters is large.
This suggests that the filter does not scale very well with the system dimension. 
A possible solution to this problem is the Rao-Blackwellization method \cite{murphy2001rao}, in which some variables are marginalized out so that the filter only needs to simulate a model of reduced dimensionality. 
For mass-action networks, such a scalable filter has been provided in \cite{zechner2014scalable}, where the model parameters are integrated out by assuming their priors to be Gamma distributions. 
However, the theory for general kinetics is still undeveloped, and the combination of it with the RPF also remains to be investigated. 
Furthermore, the numerical examples indicate that the posterior of some hidden dynamic variables can be either very noisy or extremely narrow disregarding the model parameters  (see the discussion in \Cref{subsection gene transcription model} and \Cref{subsection transcription regulation network}). 
This means that these variables are either impossible to infer accurately or can be learned in a more clever way with biological insights into the system. 
Consequently, eliminating these variables can be another option to further reduce the system and improve the filter.
Lastly, it is known that unweighted particle filters can avoid the curse of dimensionality in some circumstances \cite{surace2019avoid}; therefore, it is worth applying this methodology to biological filtering problems. 
However, conventional unweighted particle filters (e.g., the particle flow method \cite{daum2008particle} and feedback particle filters \cite{yang2013feedback}) are designed for systems of continuous state spaces. 
For SRNs whose state space is discrete, one needs to significantly modify the theory and algorithms.

\appendix

\section{The proof of \cref{thm main}}\label{section proof}

We first outline the strategy to prove our main result.
Note that the error of a particle filter comes from model reduction, finite sampling, resampling, and artificial evolution.
Here, we can easily analyze the error introduced by model reduction using the convergence of reduced models and investigate the error caused by finite sampling and resampling using the tools in  \cite{del2000branching,crisan2001particle,bain2008fundamentals}.
Therefore, the key to proving our main result lies in the analysis of artificial evolution. 
In this appendix, we show that the error caused by artificial evolution is of the order $\frac{1}{\sqrt{M}}$ (where $M$ is the particle population) by the the parameter sensitivity analysis, and, thus, this error vanishes as the particle population goes to infinity. 
Finally, by combining all the error analysis results, we can prove the main theorem. 

To make the idea more explicit, we formulate the proof outline as follows.
First, since the model reduction is not the key factor in this proof, we simplify the notation by removing the superscript that specifies the dynamical model applied.
Specifically, for a given RPF constructed in \Cref{def RPFs for SCRNs}, we term
\begin{equation}\label{eq. notations}
	\begin{array}{lcl}
	\bar \pi_{M,t_i}(\phi) & \text{as} & 
	 \text{the particle filter }\bar \pi^{N,\gamma_1}_{M,t_i}(\phi),~ \bar \pi^{N,\gamma_1,\tau}_{M,t_i}(\phi), ~\text{or} ~\bar \pi^{N,\gamma_1,H}_{M,t_i}(\phi), \\
	 \hat \pi_{M,t_i}(\phi) &\text{as} & \text{the estimate of the function $\phi$ based on resampled particles} \\
	 && (\hat \pi_{M,t_i}(\phi) \triangleq \frac{1}{M}{\sum_{j=1}^{M} \phi(\hat \kappa_j(t_i), \bar x_j(t_i))}),\\
	 \tilde \pi_{M,t_i}(\phi) &\text{as} &\text{the estimate of the function $\phi$ based on artificially perturbed} \\
	 && \text{particles } (\tilde \pi_{M,t_i}(\phi) \triangleq \frac{1}{M}\sum_{j=1}^{M}  \phi(\bar \kappa_j(t_i), \bar x_j(t_i))),\\
	 \pi_{t_{i+1}}(\phi) &\text{as} & \text{the exact filter $\pi^{N,\gamma_{1}}_{t_{i+1}}(\phi)$},\\
	 K_t & \text{as} & \text{the transition kernel of the applied dynamical system} \\ &&\text{(the CTMC, tau-leaping model, or PDMP)},\\
	 K^{N,\gamma_1}_t & \text{as} & \text{the transition kernel of the full model}\\
	 Y(t_i) &\text{as} & \text{the observation $Y^{N,\gamma_{1}}(t_{i})$},\\
	 \Delta t_i &\text{as} & t_{i+1}-t_{i}.
	\end{array}
\end{equation}
With these notations, we can decompose the error between the RPF $\bar \pi_{M,t_{i+1}} (\phi)$ and the exact filter $\pi_{t_{i+1}}(\phi)$ as follows.
\begin{align}\label{eq. bias caused by artificial evolution}
	\bar \pi_{M,t_{i+1}} (\phi) - \pi_{t_{i+1}}(\phi) 
	&= \bar \pi_{M,t_{i+1}} (\phi)  - \frac{\pi_{t_{i}}( K^{N,\gamma_1}_{\Delta t_i} \ell _{Y(t_{i+1})}\phi)}{\pi_{t_{i}}( K^{N,\gamma_1}_{\Delta t_i} \ell _{Y(t_{i+1})})} \notag\\
	&=\underbrace{\left(\bar \pi_{M,t_{i+1}} (\phi)- \frac{\tilde \pi_{M,t_{i}} ( K_{\Delta t_i} \ell _{Y(t_{i+1})}\phi)}{Z_{t}} \right)}_\text{the error caused by finite sampling}  \\
	&\quad+\frac{1}{Z_t} 
	\underbrace{\Big({\tilde \pi_{M,t_{i}} ( K_{\Delta t_i} \ell _{Y(t_{i+1})}\phi)} -
	{\hat \pi_{M,t_{i}} ( K_{\Delta t_i} \ell _{Y(t_{i+1})}\phi)} \Big)}_\text{the error caused by the artificial evolution}
	\notag\\
	&\quad+\frac{1}{Z_t}
		\underbrace{\Big(
		{\hat \pi_{M,t_{i}} ( K_{\Delta t_i} \ell _{Y(t_{i+1})}\phi)}
	- {\bar \pi_{M,t_{i}} ( K_{\Delta t_i} \ell _{Y(t_{i+1})}\phi)}
	\Big)}_\text{the error caused by resampling} \notag\\
	&\quad+\frac{1}{Z_t}
    \underbrace{\Big(
	{\bar \pi_{M,t_{i}} ( K_{\Delta t_i} \ell _{Y(t_{i+1})}\phi)}
	-  { \bar \pi_{M,t_{i}}( K^{N,\gamma_1}_{\Delta t_i} \ell _{Y(t_{i+1})}\phi)}
	\Big)}_\text{the error caused by model reduction} \notag  \\
		&\quad+\frac{1}{Z_t} 
	\underbrace{\Big(
	{ \bar \pi_{M,t_{i}} ( K^{N,\gamma_1}_{\Delta t_i} \ell _{Y(t_{i+1})}\phi)}
	-  {\pi_{t_{i}}( K^{N,\gamma_1}_{\Delta t_i} \ell _{Y(t_{i+1})}\phi)}
	\Big)}_\text{the error of the particle filter at $t_i$} \notag
\end{align}
where the first equality follows from the recursive formulas \eqref{eq. recurssive expression 1} and \eqref{eq. recurssive expression 2}, the term $\ell_{Y(t_{i+1})} $ is the likelihood function (see \eqref{eq. recurssive expression 2} and the text below), and $Z_{t}=\pi_{t_{i}}( K^{N,\gamma_1}_{\Delta t_i} \ell _{Y(t_{i+1})})$.
We note that for given observations, both the likelihood function and $Z_t$ are lower bounded by a positive constant due to the boundedness of $h(\cdot)$ (see \cref{subsection filtering problems}), so $\frac{1}{Z_t}$ is uniformly upper bounded and, therefore, not a concern in the error analysis.
\eqref{eq. bias caused by artificial evolution} tells that the error analysis of a particle filter can be performed in a recursive fashion where one analyzes the error of the filter at time $t_{i+1}$ after estimating the error at time $t_i$.
Therefore, the key to proving our result lies in investigating the first four terms on the right hand side of this equality, which corresponds to finite sampling, artificial evolution, resampling, and model reduction, respectively. 
We also note that the first and the third terms have already been analyzed in the classical literature \cite{del2000branching,crisan2001particle,bain2008fundamentals}, and the fourth can be analyzed by the convergence of the reduced model.
Consequently, the most challenging part is to analyze the second term, which corresponds to the error caused by the artificial evolution step.

Technically, we estimate the second moment of the error caused by artificial evolution.	
Notice that for any function $\phi$ satisfying \eqref{eq. sensitivity condition of phi}, we can find a positive constant $\tilde C_{\phi}$ such that
 \begin{align} \label{eq. tilde hat convergence}
&\mathbb E\left[ \left| \tilde \pi_{M,t_{i+1}}(\phi) - \hat \pi_{M,t_{i+1}}(\phi)\right|^2 \left| \mathcal Y_{t_{i+1}} \right. \right] \\
&= \mathbb E\left[ \left. \left|  \frac{1}{M}\sum_{j=1}^M \phi(\hat \kappa_j(t_{i+1}), \bar x_j(t_{i+1})) - \phi(\bar \kappa_j(t_{i+1}), \bar x_j(t_{i+1}))\right|^2 \right| \mathcal Y_{t_{i+1}}  \right] \notag\\
&\leq \frac{\tilde C^2_{\phi} }{M} \sum_{j=1}^{M} \mathbb E\left[  \|\bar \kappa_{j}(t_{i+1}) - \hat \kappa_{j}(t_{i+1}) \|^2_2 \left(1+\|\bar x_{j}(t_{i+1}) \|^q_2\right)^2 \big| \mathcal Y_{t_{i+1}}\right] \notag\\
&\leq \frac{\tilde C^2_{\phi}C^2_{\eta} \mathcal V}{M} \sum_{j=1}^{M}\frac{1}{M}\mathbb E\left[ \left(1+\|\bar x_{j}(t_{i+1})\|^q_2\right)^2\big| \mathcal Y_{t_{i+1}}  \right] \notag
\end{align}
where the third line follows from the midpoint theorem and the Jensen inequality, and the last line follows from \Cref{cond. conditions on the artificial noise 1}.
Therefore, we can analyze the effect of artificial evolution by checking whether $K_{\Delta t_i} \ell _{Y(t_{i+1})}\phi$ satisfies \eqref{eq. sensitivity condition of phi} and estimating the conditional moments of particles $\sum_{j=1}^{M}\frac{1}{M}\mathbb E\left[ \left(1+\|\bar x_{j}(t_{i+1})\|^q_2\right)^2\big| \mathcal Y_{t_{i+1}}  \right]$, of which the former corresponds to the parameter sensitivity analysis of the transition kernel. 

Following this idea, in \Cref{subsection proof parameters}, we use the parameter sensitivity analysis to show that the second term on the right hand side of \eqref{eq. bias caused by artificial evolution} is of order $\frac{1}{\sqrt{M}}$.
Specifically, we first prove in \Cref{subsection proof parameters} that both the state vector and the parameter sensitivity for the transition kernel grow mildly under the assumed conditions (\cref{proposition finiteness of moments} and \cref{proposition Kphi}). 
Then, based on these results and \eqref{eq. tilde hat convergence}, we can find a $\mathcal Y_{t_{i+1}}$ measurable random variable $\tilde C_{t_{i+1}}$ such that
    \begin{align*}
\mathbb E\left[ 
\big| \text{the second term on the right hand side of }\eqref{eq. bias caused by artificial evolution} \big|^2 \big| \mathcal Y_{t_{i+1}}  \right]  
\leq \frac{\tilde C_{t_{i+1}}}{M},
\end{align*}
	which indicates that the error caused by artificial evolution is of order $\frac{1}{\sqrt{M}}$ (see \Cref{remark bias}).
	Finally, based on the obtained results, we prove the convergence of the established RPF  by mathematical induction in \Cref{subsection proof main result}.

\subsection{Parameter sensitivities for transition kernels and the analysis of the error introduced by artificial evolution} \label{subsection proof parameters}
In this subsection, we use the parameter sensitivity analysis to estimate the error introduced by the artificial evolution and show that it vanishes as the particle population goes to infinity. 
Recall that the key is to show  $K_{\Delta t_i} \ell _{Y(t_{i+1})}\phi$ satisfying \eqref{eq. sensitivity condition of phi} and moments of particles being bounded from above (see \eqref{eq. tilde hat convergence}).

First, \Cref{assumption for full model} and \Cref{assumption non-negativity of tau-leaping algorithm} suggest that any moment of the state vector is bounded, and, therefore, the moments of particles are also bounded.

\begin{proposition}[Adapted from {\cite[Lemma 5.1]{gupta2013unbiased}} and {\cite[Theorem 4.5]{rathinam2016convergence}}]\label{proposition finiteness of moments}
	If \Cref{assumption for full model} holds, then for all $p>0$, there exists a positive constant $C_p$ such that 
	\begin{align*}
	\mathbb E\left[ \left. \left\|X^{N,\gamma_1}(t)\right\|^{p}_2 ~\right| \mathcal K= \kappa, X^{N,\gamma_1}(0)=x \right] &\leq 
	\left( \left\|x\right\|^{p}_2
	+ C_p t \right) \text{e}^{C_p t}\\
	\mathbb E\left[ \left. \left\|X^{\gamma_1}(t)\right\|^{p}_2 ~\right| \mathcal K = \kappa, X^{\gamma_1}(0)=x \right] &\leq 
	\left( \left\|x\right\|^{p}_2
	+ C_p t \right) \text{e}^{C_p t}
	\end{align*}
	for any $t>0$ and $(\kappa,x) \in \Theta\times \mathbb R^{n}_{\geq 0}$.
	If moreover \Cref{assumption non-negativity of tau-leaping algorithm} holds, then for $p>0$, there exists a constant $C_{p,\tau}$ such that 
	\begin{align*}
	\sup_{t\in[0,T]} 
	\mathbb E \left[ \left. \left\| X^{N,\gamma_1}_{\tau}(t) \right\|_2^p \right|  \mathcal K=\kappa, X^{N,\gamma_1}_{\tau}(t)=x  \right]
	\leq 	\left( \left\|x\right\|^{p}_2
	+ 1\right) \text{e}^{C_{p,\tau} t }
	\end{align*} 
	for all $(\kappa,x) \in \Theta\times \mathbb R^{n}_{\geq 0}$.
\end{proposition}
\begin{proof}
	The results for the full model and the PDMP are adapted from {\cite[Lemma 5.1]{gupta2013unbiased}}; the result for the tau-leaping method is adapted from {\cite[Theorem 4.5]{rathinam2016convergence}}.
\end{proof}

Then, we look at the parameter sensitivity for the CTMC, tau-leaping model, and PDMP. 
Given an initial condition $(\kappa,x)$ and a time point $t$, we term
\begin{align*}
	\Psi^{N}_{t,g}(\kappa,x) &= \mathbb E \left[ g\left(X^{N,\gamma_1}(t)\right) \left|~\mathcal K\triangleq \kappa,  X^{N,\gamma_1}(0)=x\right.\right],  \\
	\Psi^{N,\tau}_{t,g}(\kappa,x) &= \mathbb E \left[ g\left(X^{N,\gamma_1}_{\tau}(t)\right) \left|~\mathcal K\triangleq \kappa,  X^{N}(0)=x\right.\right], \\
	\Psi_{t,g}(\kappa,x) &= \mathbb E \left[ g\left(X^{\gamma_1}(t)\right) \left|~\mathcal K\triangleq \kappa,  X^{N,\gamma_1}(0)=x\right.\right]  
\end{align*}
as the conditional expectations of function $g$ for the CTMC, tau-leaping model, and PDMP, respectively.
Then, the parameter sensitivity for these models are defined, respectively, by $\frac{\partial \Psi^N_{t,g}(\kappa,x) }{\partial \kappa}$, $\frac{\partial \Psi^{N,\tau}_{t,g}(\kappa,x) }{\partial \kappa}$, and $\frac{\partial \Psi_{t,g}(\kappa,x) }{\partial \kappa}$.
The literature \cite{gupta2018estimation} shows that the parameter sensitivity for the full model can be computed by 
\begin{small}
\begin{align}
&\frac{\partial \Psi^N_{t,g}(\kappa,x) }{\partial \kappa} \label{eq. parameter sensitivity for SCRNs} \\
&= \sum_{j=1}^{r} \mathbb{E}
\left[
\left. 
\int_{0}^{t} \frac{\partial \lambda^{N}_{j}\left(\kappa, X^{N,\gamma_1}(s)\right)}{\partial \kappa} 
\Delta^N_j \Psi^N_{t-s,g}\left(\kappa,X^{N,\gamma_1}(s)\right) 
\dd s
\right| ~\mathcal K=\kappa,  X^{N}(0)=x 
\right] \notag
\end{align}
\end{small}where $\Delta^N_j \Psi^N_{t-s,g}\left(\kappa,x\right)\triangleq \Psi^N_{t-s,g}\left(\kappa,x+\Lambda^N \zeta_j\right)- \Psi^N_{t-s,g}\left(\kappa,x\right)$.
Similarly, by mathematical induction, we can show that the parameter sensitivity for the tau-leaping model (with a deterministic time-discretization scheme) can be computed by
\footnote{Note that $\frac{\partial \Psi^{N,\tau}_{t,g}(\kappa,x) }{\partial \kappa}$ computes the parameter sensitivity when the tau-leaping model is viewed as the ground truth. It is not necessarily equal or even close to $\frac{\partial \Psi^N_{t,g}(\kappa,x) }{\partial \kappa}$, the parameter sensitivity for the full model.}
\begin{small}
\begin{align}
&\frac{\partial \Psi^{N,\tau}_{t,g}(\kappa,x) }{\partial \kappa} \label{eq parameter sensitivity for the tau-leaping method} \\
&= \sum_{j=1}^{r} \mathbb{E}
\Bigg[
\sum_{\tau_i<t}  (\tau_{i+1}\wedge t -\tau_i)\frac{\partial \lambda^{N}_{j}\left(\kappa, X^{N,\gamma_1}_\tau (\tau_i)\right)}{\partial \kappa} 
\Delta^N_j \Psi^N_{t-s,g}\left(\kappa,X^{N,\gamma_1}_{\tau}(\tau_{i+1})\right) \notag \\
& \qquad\qquad\qquad\qquad\qquad\qquad\qquad\qquad\qquad\qquad\qquad\qquad
\bigg| ~\mathcal K=\kappa,  X^{N,\gamma_1}_{\tau}(0)=x 
\bigg]. \notag
\end{align}
\end{small}Finally, the literature \cite{gupta2019sensitivity} tells that under \Cref{assumption convergence of the initial condition}, \Cref{assumption for full model}, and \Cref{assumption for PDMPs}, the parameter sensitivity for the PDMP satisfies
\begin{equation}\label{eq. parameter sensitivity PDMP}
	\frac{\partial \Psi_{t,g}(\kappa,x) }{\partial \kappa} = \lim_{N\to\infty} \frac{\partial \Psi^N_{t,g}(\kappa,x) }{\partial \kappa},
\end{equation}
if $g$ is bounded, continuously differentiable.
Using these results, we can further prove that the parameter sensitivity for the transition kernel is at most polynomially growing with respect to the state argument (see \Cref{proposition Kphi}).

\begin{proposition}(Polynomially growing rate of the parameter sensitivity for the transition kernel)\label{proposition Kphi}
	For any measurable function $\phi$, we denote 
	\begin{align*}
		\left(K^{N,\gamma_1}_t \phi \right)(\kappa, x) &\triangleq \mathbb E_{\mathbb P} \left[ \phi\left(\mathcal K, X^{N,\gamma_1}(t)\right) \left|~\mathcal K=\kappa,  X^{N}(0)=x\right.\right], \\
		\left(K^{N,\gamma_1,\tau}_t \phi \right)(\kappa, x) &\triangleq \mathbb E_{\mathbb P} \left[ \phi\left(\mathcal K, X^{N,\gamma_1}_{\tau}(t)\right) \left|~\mathcal K=\kappa,  X^{N}(0)=x\right.\right], \\
		\left(K^{\gamma_1}_t \phi \right)(\kappa, x) &\triangleq \mathbb E_{\mathbb P} \left[ \phi\left(\mathcal K, X^{\gamma_1}(t)\right) \left|~\mathcal K=\kappa,  X^{\gamma_1}(0)=x\right.\right].
	\end{align*}
	If $\phi:\Theta\times \mathbb R^{n}_{\geq 0} \to \mathbb R$ 
	satisfies the requirement in \cref{thm main}, then the following results hold.
	\begin{itemize}
		\item If \Cref{assumption for full model} holds, then for any $t>0$ there exists a constant $C_{t,\phi}$ such that $\left\| \frac{ \partial \left(K^{N,\gamma_1}_t \phi \right)(\kappa, x) }{\partial \kappa}\right\|_{2} 
		\leq C_{t,\phi} \left(1+ \|x\|^q_2\right)$ for all $ (\kappa,x) \in \Theta\times \mathbb R^{n}_{\geq 0}$.
		\item If \Cref{assumption non-negativity of tau-leaping algorithm} and \Cref{assumption for full model} hold, then for any $t>0$ there exists a constant $C_{t,\phi,\tau}$ such that $\left\| \frac{ \partial \left(K^{N,\gamma_1,\tau}_t \phi \right)(\kappa, x) }{\partial \kappa}\right\|_{2} 
		\leq C_{t,\phi,\tau} \left(1+ \|x\|^q_2\right)$ for all $ (\kappa,x) \in \Theta\times \mathbb R^{n}_{\geq 0}$.
		\item If \Cref{assumption convergence of the initial condition}, \Cref{assumption for full model}, and \Cref{assumption for PDMPs} hold, then for any $t>0$ and $ (\kappa,x) \in \Theta\times \mathbb R^{n}_{\geq 0}$ there holds  $	\left\| \frac{ \partial \left(K^{\gamma_1}_t \phi \right)(\kappa, x) }{\partial \kappa}\right\|_{2} \leq C_{t,\phi} \left(1+ \|x\|^q_2\right)$,
		where $C_{t,\phi}$ is the same as that in the first result.
	\end{itemize}
\end{proposition}

\begin{proof}
    We first prove the result for the transition kernel of the full model.
    By definition, we can write the derivative of the transition kernel by 
    \begin{align}
    \frac{ \partial K^{N,\gamma_1}_t \phi (\kappa, x) }{\partial \kappa} \notag
    =& \mathbb E \left[ \left. \frac{\partial \phi\left(\kappa, X^{N,\gamma_1}(t)\right)}{\partial \kappa} \right|~\mathcal K=\kappa,  X^{\gamma_1}(0)=x \right]
    +
    \left.\frac{\partial \Psi^{N}_{t,\phi(\kappa, \cdot)}(\theta,x) }{\partial \theta}\right|_{\theta=\kappa}. 
    \end{align}
    By \cref{eq. sensitivity condition of phi} and \cref{proposition finiteness of moments}, we can find a positive constant $C^1_{t,\phi}$ such that 
    \begin{equation}\label{eq. sensitivity analysis help 1}
    \mathbb E \left[ \left. \left\| \frac{\partial \phi\left(\kappa, X^{N,\gamma_1}(t)\right)}{\partial \kappa}\right\|_2 \right|~\mathcal K=\kappa,  X^{\gamma_1}(0)=x \right]
    \leq C^1_{t,\phi}\left( 1
    + \|x\|^q_2 \right)
    \end{equation}
    for all $(\kappa,x) \in \Theta\times \mathbb R^{n}_{\geq 0}$.
    Also, by \eqref{eq. parameter sensitivity for SCRNs} and the boundedness of $\phi$, there holds
    \begin{align}
    \left\| \left.\frac{\partial \Psi^{N}_{t,\phi(\kappa, \cdot)}(\theta,x) }{\partial \theta}\right|_{\theta=\kappa}  \right\|_2
    &\leq  2\|\phi\|_{\infty}\sum_{j=1}^{r} \left\|\mathbb{E}
    \left[
    \left. 
    \int_{0}^{t} \frac{\partial \lambda^{N}_{j}\left(\kappa, X^{N,\gamma_1}(s)\right)}{\partial \kappa} 
    \dd s
    \right| ~\mathcal K=\kappa,  X^{N}(0)=x 
    \right]\right\|_2 \notag \\
    & \leq  2\|\phi\|_{\infty}\sum_{j=1}^{r} \mathbb{E}
    \left[
    \left. 
    \int_{0}^{t} \left\| \frac{\partial \lambda^{N}_{j}\left(\kappa, X^{N,\gamma_1}(s)\right)}{\partial \kappa} \right\|_2
    \dd s
    \right| ~\mathcal K=\kappa,  X^{N}(0)=x 
    \right] \notag 
    \end{align}
    where the second inequality follows from Jensen's inequality and the triangle inequality.
    Therefore by \Cref{proposition finiteness of moments} and \Cref{assumption for full model}, we can find a constant $C^2_{t,\phi}$ such that 
    \begin{equation}\label{eq. sensitivity analysis help 2}
    \left\| \left.\frac{\partial \Psi^{N}_{t,\phi(\kappa, \cdot)}(\theta,x) }{\partial \theta}\right|_{\theta=\kappa}\right\|_2
    \leq  C^2_{t,\phi}\left( 1
    + \|x\|^q_2 \right)
    \qquad \forall (\kappa,x) \in \Theta\times \mathbb R^{n}_{\geq 0}.
    \end{equation}
    Finally, by combing \eqref{eq. sensitivity analysis help 1} and \eqref{eq. sensitivity analysis help 2}, we prove the result for the transition kernel of the full model (the first result).
    
    By the same argument, we can also prove the second result. 
    For the third result, we can learn from \eqref{eq. parameter sensitivity PDMP} that $\Psi_{t,\phi(\kappa, \cdot)}(\theta,x)$ satisfies \eqref{eq. sensitivity analysis help 2}.
    Moreover, by \cref{eq. sensitivity condition of phi} and \cref{proposition finiteness of moments}, we can show that the quantity $\phi\left(\kappa, X^{\gamma_1}(t)\right)$ also satisfies \eqref{eq. sensitivity analysis help 1}.
    Therefore, the third result is proven. 
\end{proof}

\begin{remark}[Analysis of the error introduced by artificial evolution]{\label{remark bias}}
	Now, we use the obtained propositions to show that the second term on the right hand side of \eqref{eq. bias caused by artificial evolution} is of the order $1/\sqrt{M}$.
    First, we look at the conditional moments of the particles in the RPF.
    We note that every weight $w_{j}(t_i)$ ($j=1,\dots,M$, $i=1,2,\dots,$) is upper bounded by the constant $\frac{1}{M}\exp\left(  \left\|Y(t_i)\right\|^2_2+ m\|h\|^2_{\infty}\right)$, and the resampling does not change the expectation of the empirical mean of these weighted particles. 
    These facts together with the third claim of \cref{assumption for full model} and  \cref{proposition finiteness of moments} suggests that given the observation process, the second conditional moment of the particles is bounded from above, i.e., 
    \begin{align}\label{eq. boundedness of particles' moments}
    	\frac{1}{M} \sum_{j=1}^n \mathbb E\left[ \left(1+\|\bar x_{j}(t_{i+1})\|^q_2\right)^2\big| \mathcal Y_{t_i+1} \right] \leq \tilde C_{Y(t_{i+1})} 
    \end{align}
    where $\tilde C_{Y(t_{i+1})}$ is a $\mathcal Y_{t_i+1} $ measurable random variable.
    Second, we note that  if $\phi$ satisfies the requirement in  \Cref{thm main}, the function $\ell_{y} \phi$ (for any $y\in\mathbb R^{m}$) also satisfies this requirement, because the likelihood function $\ell_{y}$ is bounded and has no dependence on the argument $\kappa$.
    Consequently, by \cref{proposition Kphi}, the function $K_{\Delta t_i} \ell _{Y(t_{i+1})}\phi$ also satisfies the requirement in  \Cref{thm main} under the proposed conditions.
    Finally, applying these results to \eqref{eq. tilde hat convergence}, we can find a $\mathcal Y_{t_i+1} $ measurable random variable $\tilde C_{t_{i+1}}$ such that 
     \begin{align}\label{eq error analysis the second term}
    \mathbb E\left[ 
    \big| \text{the 2nd term on the right hand side of }\eqref{eq. bias caused by artificial evolution} \big|^2 \big| \mathcal Y_{t_i+1}  \right]  
    \leq \frac{\tilde C_{t_{i+1}}}{M},
    \end{align}
    which suggests that the bias contributed by artificial evolution is of order $\frac{1}{\sqrt{M}}$.

\end{remark}

\subsection{Proving the main result by mathematical induction} {\label{subsection proof main result}}
In this subsection, we prove \cref{thm main} using the results obtained in the previous section.

\begin{proof}[The Proof of \cref{thm main}]
   For simplicity, we use the simplified notations in \eqref{eq. notations} to prove the result.
   Moreover, we introduce another quantity $a$ to indicate the fidelity of the reduced model.
   Specifically, if the RPF uses the tau-leaping model, then $a=|\tau|$; 
   if the RPF uses the PDMP model, then $a=1/N$;
   if the RPF uses the full model, then $a\equiv 0$.
   By the dominated convergence theorem, to prove the result, we only need to show 
   \begin{align}\label{eq. alternative quantity}
   	  \lim_{M\to\infty} \lim_{a\to\infty}  \mathbb E\left[ 	\left(\bar \pi_{M,t_{i}} (\phi) - \pi_{t_{i}}(\phi)\right)^2 \big| \mathcal Y_{t_{i}}  \right] = 0
   \end{align}	
   for any $i\in\mathbb N_{\geq 0}$ and any $\phi$ satisfying the requirement in \cref{thm main}.
   Therefore, in the sequel, we focus on analyzing the conditional error of the filters. 
   
   We use mathematical induction to prove \eqref{eq. alternative quantity}.
   For $i=1$, the relation \eqref{eq. alternative quantity} holds automatically by the law of large numbers and the convergence of the reduced models. 
   Then, we are going to show that if \eqref{eq. alternative quantity} holds for a positive integer $i$, then it also holds for the positive integer $i+1$.
   Note that we can decompose the error of the filter as \eqref{eq. bias caused by artificial evolution}; therefore, the key is to separately analyze the error caused by finite sampling, artificial evolution, resampling, and model reduction. 
   
   \textit{Error caused by model reduction:}
   Recall that for given observations, both the likelihood function and $Z_t$ in \eqref{eq. bias caused by artificial evolution} are lower bounded by a positive constant due to the boundedness of $h(\cdot)$ (see \cref{subsection filtering problems}), so $\frac{1}{Z_t}$ is uniformly upper bounded and, therefore, not a concern in the error analysis.
   Then, by the convergence of the reduced models (\cref{assumption convergence of tau-leaping algorithm} and \cref{prop Kang kurtz gamma1}) and the conditional dominated convergence theorem, we have that for any bounded continuous function $\phi$, there holds
   \begin{align}\label{eq. error 4th}
   	   \lim_{a\to0} \mathbb E\left[ \left( \text{the 4th term on the RHS of \eqref{eq. bias caused by artificial evolution}}\right) ^2 \big | \mathcal Y_{t_{i+1}}\right] =0.
   \end{align} 
   
   \textit{Error caused by resampling:}
   	If we use the multinomial resampling, then all the resampled particles are conditionally independently and identically distributed given the observation process, and, therefore,  we can conclude \cite{crisan2001particle}
   \begin{equation*}
   \mathbb E \left[ \left. \left| \hat\pi_{M,t_{i}}\left(K_{\Delta t_i} \ell _{Y(t_{i+1})}\phi\right) - \bar \pi_{M,t_{i+1}}\left(K_{\Delta t_i} \ell _{Y(t_{i+1})}\phi\right)\right|^2\right| \mathcal Y_{t_{i+1}} \right] \leq 
   \frac{ \|\ell _{Y(t_{i+1})}\phi\|^2_{\infty}}{M},
   \end{equation*}
   where $\|\ell _{Y(t_{i+1})}\phi\|_{\infty}$ is the uniform norm of the function $\ell _{Y(t_{i+1})}\phi$.
   Moreover, by \cite[Exerice 9.1]{bain2008fundamentals} which states that residual resampling introduces the minimum noise among all resampling methods, we can also conclude the above relation for residual resampling, and, therefore, 
   \begin{align}\label{eq. error 3rd}
    \lim_{M\to\infty} \lim_{a\to\infty} \mathbb E\left[ \left( \text{the 3rd term on the RHS of \eqref{eq. bias caused by artificial evolution}}\right) ^2 \big | \mathcal Y_{t_{i+1}}\right] =0,
   \end{align} 
   for any bounded $\phi$.
   
   \textit{Error caused by artificial evolution:} By \eqref{eq error analysis the second term} in \cref{remark artificial evolution}, we can conclude that for any $\phi$ satisfies the requirement in \Cref{thm main}, there holds
   \begin{align}\label{eq. error 2nd}
    \lim_{M\to\infty} \lim_{a\to\infty} \mathbb E\left[ \left( \text{the 2nd term on the RHS of \eqref{eq. bias caused by artificial evolution}}\right) ^2 \big | \mathcal Y_{t_{i+1}}\right] =0.
   \end{align} 
   
   \textit{Error caused by finite sampling:}
   For any $\phi$ satisfying the requirement in \cref{thm main}, the relations \eqref{eq. error 4th}, \eqref{eq. error 3rd}, \eqref{eq. error 2nd}, and \eqref{eq. alternative quantity} (for the integer $i$) imply
   \begin{align}\label{eq. error 1st help}
   	    \lim_{M\to\infty} \lim_{a\to\infty}  \mathbb E \left[ \left. \left| \tilde \pi_{M,t_{i}}\left(K_{\Delta t_i} \ell _{Y(t_{i+1})}\phi\right) -  \pi_{M,t_{i+1}}\left(K^{N,\gamma_1}_{\Delta t_i} \ell _{Y(t_{i+1})}\phi\right)\right|^2\right| \mathcal Y_{t_{i+1}} \right]=0.
   \end{align}
   Therefore, to analyze the error caused by finite sampling, we only need to compare $\bar \pi_{M,t_{i}}\left(\phi\right) $ and $\frac{\tilde  \pi_{M,t_{i+1}}\left(K_{\Delta t_i} \ell _{Y(t_{i+1})}\phi\right)}{\tilde \pi_{M,t_{i+1}}\left(K_{\Delta t_i} \ell _{Y(t_{i+1})}\right)}$.
   	For any function $\phi$ satisfying the requirement in \Cref{thm main},  we define
   	\begin{align*}
   	\bar \rho_{M,t_{i+1}} (\phi)  = {\sum_{j=1}^{M} \phi(\kappa_j(t_{i+1}), x_j(t_{i+1}))} \qquad \forall i\in\mathbb N,
   	\end{align*} 
   	which is the estimate of the function $\phi$ right before the collection of $Y(t_{i+1})$.
   	Since the increments of the particles in the sampling step are independent, this quantity satisfies
   	\begin{align}\label{eq. estimate bar rho}
   	&\mathbb E \left[ \left. \left| \bar \rho_{M,t_i+1}(\phi) - \tilde \pi_{M,t_{i+1}}\left(K^{N,\gamma_1}_{\Delta t_i} \ell _{Y(t_{i+1})}\phi\right) \right|^2 \right| \mathcal Y_{t_{i+1}} \right]  \leq \frac{\|\ell _{Y(t_{i+1})}\phi\|^2_{\infty}}{M} .
   	\end{align}
   	Moreover, by \cite[(2.10)]{crisan2001particle}, there is
   	\begin{align*}
   	&\left|\bar\pi_{M,t_{i+1}}(\phi) - \frac{\tilde  \pi_{M,t_{i+1}}\left(K_{\Delta t_i} \ell _{Y(t_{i+1})}\phi\right)}{\tilde \pi_{M,t_{i+1}}\left(K_{\Delta t_i} \ell _{Y(t_{i+1})}\right)}\right| \\
   	&\leq
   	\frac{\left\|\phi\right\|_{\infty}}{\tilde \pi_{M,t_{i+1}}\left(K_{\Delta t_i} \ell _{Y(t_{i+1})}\right)}
   	\left|\bar  \rho_{M,t_i+1}\left( \ell_{y_{i+1}}\right)
   	- 
   	\tilde \pi_{M,t_{i+1}}\left(K_{\Delta t_i} \ell _{Y(t_{i+1})}\right)\right| \\
   	& \quad~ +
   	\frac{1}{\tilde \pi_{M,t_{i+1}}\left(K_{\Delta t_i} \ell _{Y(t_{i+1})}\right)}
   	\left|\bar \rho_{M,t_{i+1}}\left( \ell_{y_{i+1}}\phi\right) 
   	- 
   	\tilde \pi_{M,t_{i+1}}\left(K_{\Delta t_i} \ell _{Y(t_{i+1})}\phi\right)\right|.
   	\end{align*}
   	Recall that the likelihood function $\ell _{Y(t_{i+1})}$ has a $\mathcal Y_{t_{i+1}}$ measurable, positive lower bound due to the boundedness of $h(\cdot)$, so $\tilde \pi_{M,t_{i+1}}\left(K_{\Delta t_i} \ell _{Y(t_{i+1})}\right)$ is also greater than or equal to this bound.
   	Thus, by applying \eqref{eq. estimate bar rho} to this inequality, we have
   	\begin{equation*}
   		\lim_{M\to\infty} \lim_{a\to0}\mathbb E \left[ 
   		\left(\bar\pi_{M,t_{i+1}}(\phi) - \frac{\tilde  \pi_{M,t_{i+1}}\left(K_{\Delta t_i} \ell _{Y(t_{i+1})}\phi\right)}{\tilde \pi_{M,t_{i+1}}\left(K_{\Delta t_i} \ell _{Y(t_{i+1})}\right)}\right)^2
   		\Bigg|
   		\mathcal Y_{t_{i+1}}
   		\right]
   		=0,
   	\end{equation*} 
   	which together with \eqref{eq. error 1st help} suggest the error caused by finite sampling is controlled, i.e., 
   	 \begin{align}\label{eq. error 1st}
   	\lim_{M\to\infty} \lim_{a\to\infty} \mathbb E\left[ \left( \text{the 1st term on the RHS of \eqref{eq. bias caused by artificial evolution}}\right) ^2 \big | \mathcal Y_{t_{i+1}}\right] =0.
   	\end{align} 
   	
   	Finally, by applying \eqref{eq. error 4th}, \eqref{eq. error 3rd}, \eqref{eq. error 2nd}, and \eqref{eq. error 1st} to \eqref{eq. bias caused by artificial evolution}, we can conclude that if \eqref{eq. alternative quantity} holds for a positive integer $i$, then it also holds for the positive integer $i+1$, which proves the result.
\end{proof}


\end{document}